\documentclass[a4paper,USenglish]{lipics-v2016}

\newcommand{\lv}[1]{#1}
\newcommand{\sv}[1]{}

\usepackage{boxedminipage}
\usepackage{cite}
\usepackage{color}
\usepackage[linesnumbered,boxed]{algorithm2e}
\usepackage{graphicx}
\usepackage{latexsym}
\usepackage{amsmath,verbatim}
\usepackage{amssymb}
\usepackage{enumerate}
\usepackage{todonotes}
\usepackage[draft,author=]{fixme}
\fxsetup{theme=color}

\newtheorem{THE}{Theorem}

\newtheorem{LEM}{Lemma}

\newtheorem{ourfact}{Fact}

\newcommand{\tw}{\text{\normalfont{\bfseries tw}}}
\newcommand{\fvs}{\text{\normalfont{\bfseries fvs}}}
\newcommand{\und}[1]{\overline{#1}}

\newcommand{\PP}{\mathcal{P}}

\newcommand{\QQ}{\mathcal{Q}}

\newcommand{\cF}{\mathcal{F}}

\newcommand{\SSS}{\mathcal{S}}

\newcommand{\Card}[1]{|#1|}
\newcommand{\Nat}{\mathbb{N}}

\newcommand{\bigoh}{\mathcal{O}}

\newcommand{\cc}[1]{{\mbox{\textnormal{\textsf{#1}}}}\xspace}

\newcommand{\NP}{\cc{NP}}
\newcommand{\FPT}{\cc{FPT}}
\newcommand{\XP}{\cc{XP}}
\newcommand{\Weft}{{\cc{W}}}
\newcommand{\W}[1]{{\Weft}{\normalfont{[#1]}}}

\newcommand{\hy}{\hbox{-}\nobreak\hskip0pt}

\newcommand{\concat}{\circ}

\newcommand{\bigunion}[1]{\widetilde{#1}}
\newcommand{\III}{\mathcal{I}}

\newcommand{\SB}{\{\,}%
\newcommand{\SM}{\;{|}\;}%
\newcommand{\SE}{\,\}}%

\newcommand{\pbDef}[3]{%
\noindent
\begin{center}
\begin{boxedminipage}{0.98 \columnwidth}
#1\\[5pt]
\begin{tabular}{l p{0.70 \columnwidth}}
Input: & #2\\
Question: & #3
\end{tabular}
\end{boxedminipage}
\end{center}
}

\newcommand{\pbDefP}[4]{%
\noindent
\begin{center}
\begin{boxedminipage}{0.98 \columnwidth}
#1\\[5pt]
\begin{tabular}{l p{0.70 \columnwidth}}
Input: & #2\\
Parameter: & #3\\
Question: & #4
\end{tabular}
\end{boxedminipage}
\end{center}
}

\usepackage{enumitem}
\begin{document}

\title{On Structural Parameterizations of the Edge Disjoint Paths Problem}

\author{Robert Ganian}
\author{Sebastian Ordyniak}
\author{M. S. Ramanujan}

\affil{Algorithms and Complexity group, TU Wien, Vienna, Austria \texttt{\{ganian,ordyniak,ramanujan\}@ac.tuwien.ac.at}}

\authorrunning{R. Ganian, S. Ordyniak, and M.S. Ramanujan}
\titlerunning{\sv{On Structural Paramerizations of the Edge Disjoint
    Paths Problem}\lv{On Structural Paramerizations of EDP}}

\maketitle
\date{ }

\begin{abstract}
  In this paper we revisit the classical Edge Disjoint Paths (EDP) problem,
  where one is given an undirected graph $G$ and a set of terminal
  pairs $P$ and asks whether $G$ contains a set of pairwise
  edge-disjoint paths connecting every terminal pair in $P$. Our focus
  lies on structural parameterizations for the problem that allow for
  efficient (polynomial-time or fpt) algorithms. As our
  first result, we answer an open question stated in Fleszar, Mnich,
  and Spoerhase (2016), by showing that the problem can be solved in
  polynomial time if the input graph has a feedback vertex set of size
  one. We also show that EDP parameterized by the treewidth and the
  maximum degree of the input graph is fixed-parameter tractable.

  Having developed two novel algorithms for EDP using 
  structural restrictions on the input graph, we then turn our
  attention towards the augmented graph, i.e., the graph obtained
  from the input graph after adding one edge between every terminal
  pair. In constrast to the input graph, where EDP is known to remain
  NP-hard even for treewidth two, a result by Zhou et al. (2000) shows
  that EDP can be solved in non-uniform polynomial time if the augmented graph has
  constant treewidth; we note that the possible improvement of this result to an 
  fpt-algorithm has remained open since then. We show that this is 
  highly unlikely by establishing the \W{1}-hardness of the problem 
  parameterized by the treewidth (and even feedback vertex set) of the augmented graph. Finally, 
  we develop an fpt-algorithm for EDP by exploiting a novel structural 
  parameter of the augmented graph. 

\end{abstract}

\section{Introduction}
\label{sec:intro}

The {\sc Edge Disjoint Paths} ({\sc EDP}) and {\sc Node Disjoint Paths} ({\sc NDP}) are fundamental routing graph problems. In the {\sc EDP} ({\sc NDP}) problem the input is a graph $G$, and a set $P$ containing $k$ pairs of vertices and the objective is to decide whether there is a set of $k$ pairwise edge disjoint (respectively vertex disjoint) paths connecting each pair in $P$.
These problems and their optimization versions -- {\sc MaxEDP} and
{\sc MaxNDP} -- have been at the center of numerous results in structural
graph theory, approximation algorithms, and parameterized
algorithms~\cite{RobertsonS95b,KawarabayashiKK14,ChekuriKS06,KolliopoulosS04,EneMPR16,ZhouTN00,NishizekiVZ01,GargVY97,FleszarMS16}.

When $k$ is a part of the
input, both EDP and NDP are known to be NP-complete~\cite{Karp75}. Robertson and Seymour's seminal work in the Graph Minors project~\cite{RobertsonS95b} provides an $\bigoh(n^3)$ time algorithm for both problems for every fixed value of $k$. In the realm of  {\em Parameterized Complexity}, their result can be interpreted as {\em fixed-parameter} algorithms for EDP and NDP   parameterized by $k$.
Here, one considers problems associated with a certain numerical parameter $k$ and the central question is whether the problem can be solved in time $f(k)\cdot n^{\bigoh(1)}$ where $f$ is a computable function and $n$ the input size; algorithms with running time of this form are called fpt-algorithms~\cite{FlumGrohe06,DowneyFellows13,CyganFKLMPPS15}.


While the aforementioned research considered the number of paths to be the parameter, 
another line of research investigates the effect of \emph{structural parameters} of the input graphs on the complexity of these problems. 
Fleszar, Mnich, and Spoerhase~\cite{FleszarMS16} initiated the study of NDP and EDP parameterized by the feedback vertex set number (the size of the smallest feedback vertex set) of the input graph and showed that 
EDP remains NP-hard even on graphs with  feedback vertex set number two. 
Since EDP is known to be polynomial time solvable on forests~\cite{GargVY97}, this left only the case of feedback vertex set number one open, which they conjectured to be polynomial time solvable. Our first result is a positive resolution of their conjecture. 

\begin{THE}\label{thm:main_theorem}
  {\sc EDP} can be solved in time 
  $\bigoh(|P||V(G)|^{\frac{5}{2}})$ on graphs with feedback vertex set number one.
\end{THE}

A key observation behind the polynomial-time algorithm is that an EDP instance
with a feedback vertex set $\{x\}$ is a yes-instance if and only if, for every tree $T$ of $G-\{x\}$, it is possible to connect all terminal pairs in $T$ either to each other or to $x$ through pairwise edge disjoint paths in $T$.The main ingredient of the algorithm is then a dynamic programming
procedure that determines whether such a set $S_T$ exists for a tree
$T$ of $G-\{x\}$.

Continuing to explore structural parameterizations for the input graph
of an EDP instance, we then show that even though EDP is NP-complete when the
\emph{input graph} has treewidth two, it becomes fixed-parameter tractable if
we additionally parameterize by the maximum degree.
\begin{THE}\label{thm:treewidth_theorem}
  \textsc{EDP} is fixed-parameter tractable parameterized by the
  treewidth and the maximum degree of the input graph. 
\end{THE} 

Having explored the algorithmic applications of structural restrictions on the input
graph for EDP, we then turn our attention towards similar restrictions
on the \emph{augmented graph} of an EDP instance $(G,P)$, i.e., the graph obtained
from $G$ after adding an edge between every pair of terminals in
$P$. Whereas EDP is NP-complete even if the input
graph has treewidth at most two~\cite{NishizekiVZ01}, it can be solved
in non-uniform polynomial time if the treewidth of the augmented graph is
bounded~\cite{ZhouTN00}. It has remained open whether
EDP is fixed-parameter tractable parameterized by the treewidth of the
augmented graph; interestingly, this has turned out to be the case for
the strongly related multicut problems~\cite{GottlobL07}. Surprisingly,
we show that this is not the case for EDP, by establishing the \W{1}-hardness 
of the problem parameterized by not only the treewidth but also 
by the feedback vertex set number of the augmented graph.
%
\begin{THE}\label{the:hard-edp-fvs}
  EDP is \W{1}-hard parameterized by the feedback vertex set number
  of the augmented graph.
\end{THE}
Motivated by this strong negative result, our next aim was to find
natural structural parameterizations for the augmented graph of an EDP
instance for which the problem becomes fixed-parameter tractable.
Towards this aim, we
introduce the \emph{fracture
  number}, which informally corresponds to the size of a minimum
vertex set $S$ such that the size of every component in the graph
minus $S$ is small (has size at most $|S|$). We show that EDP is
fixed-parameter tractable parameterized by this new parameter.

\begin{THE}\label{thm:fracture_theorem}
  \textsc{EDP} is fixed-parameter tractable parameterized by the fracture number of the augmented graph. 
\end{THE}
We note that the reduction
in~\cite[Theorem 6]{FleszarMS16} excludes the applicability of the
fracture number of the \emph{input graph} by showing that EDP is
NP-complete even for instances with fracture number at most three.
Finally, we complement Theorem~\ref{thm:fracture_theorem} by showing that bounding the number of
terminal pairs in each component instead of the its size is not
sufficient to obtain fixed-parameter tractability. \lv{Indeed, we show
that EDP is NP-hard even on instances $(G,P)$ where the augmented
graph $G^P$ has a deletion set $D$ of size $6$ such that every
component of $G^P \setminus D$ contains at most $1$ terminal pair.
We note that a parameter similar to the fracture number has recently
been used to obtain fpt-algorithms for Integer Linear
Programming~\cite{DEGetalIJCAI17}.}

\sv{\smallskip \noindent {\emph{Statements whose proofs are located in the appendix are marked with $\star$.}}}

\section{Preliminaries}
\label{sec:pre}

\subsection{Basic Notation} 
We use standard terminology for graph theory, see for
instance~\cite{Diestel10}.
Given a graph $G$, we let $V(G)$ denote its vertex set, $E(G)$ its
edge set and by $V(E')$ the set of vertices incident with the edges in
$E'$, where $E' \subseteq E(G)$.
The (open) neighborhood of a vertex $x \in V(G)$ is the set $\{y\in V(G):xy\in E(G)\}$ and is denoted by $N_G(x)$. For a vertex subset $X$, the neighborhood of $X$ is defined as $\bigcup_{x\in X} N_G(x) \setminus X$ and denoted by $N_G(X)$.
For a vertex set $A$, we use $G-A$ to denote the graph obtained from
$G$ by deleting all vertices in $A$, and we use $G[A]$ to denote the
\emph{subgraph induced on} $A$, i.e., $G- (V(G)\setminus A)$. A
\emph{forest} is a graph without cycles, and a vertex set $X$ is a
\emph{feedback vertex set} (\emph{FVS}) if $G-X$ is a forest. We use $[i]$ to
denote the set $\{0,1,\dots,i\}$. The \emph{feedback vertex set
  number} of a graph $G$, denoted by $\fvs(G)$, is the smallest
integer $k$ such that $G$ has a
feedback vertex set of size $k$.

\subsection{Edge Disjoint Path Problem}\label{ssec:edp}

\sv{In the \textsc{Edge Disjoint Paths (EDP)} problem, one is given an
undirected graph $G$ a set $P$ of terminal pairs (i.e., subsets of
$V(G)$ of size two) and the question is whether there is there a set
of pairwise edge disjoint paths connecting every set of terminal pairs 
in $P$.}
\lv{
  Throughout the paper we consider the following problem.
\pbDef{\textsc{Edge Disjoint Paths (EDP)}}
{A graph $G$ a set $P$ of terminal pairs, i.e., a set of subsets of $V(G)$ of size two.}
{Is there a set of pairwise edge disjoint paths connecting every set
  of terminal pairs in $P$?}
}
Let $(G,P)$ be an instance of EDP; for brevity, we will sometimes
denote a terminal pair $\{s,t\}\in P$ simply as $st$.
For a subgraph $H$ of $G$, we denote by $P(H)$ the subset of $P$
containing all sets that have a non-empty intersection with $V(H)$ and
for $P'\subseteq P$, we denote by $\bigunion{P'}$ the set $\bigcup_{p
  \in P'}p$.
 We will assume that, w.l.o.g., each vertex $v\in V(G)$ occurs in at
 most one terminal pair, each vertex in a terminal pair has degree $1$
 in $G$, and each terminal pair is not adjacent to each other; indeed,
 for any instance without these properties, we can add a new leaf
 vertex for terminal, attach it to the original terminal, and replace
 the original terminal with
 the leaf vertex~\cite{ZhouTN00}. 
\begin{definition}[\cite{ZhouTN00}]
  The \emph{augmented graph} of $(G,P)$ is the graph $G^P$ obtained
  from $G$ by adding edges between each terminal pair, i.e.,
  $G^P=(V(G), E(G)\cup P)$.
\end{definition}

\subsection{Parameterized Complexity}

\lv{
A \emph{parameterized problem} $\PP$ is a subset of $\Sigma^* \times \Nat$ for some finite alphabet $\Sigma$. Let $L\subseteq \Sigma^*$ be a classical decision problem for a finite alphabet, and let $p$ be a non-negative integer-valued function defined on $\Sigma^*$. Then $L$ \emph{parameterized by} $p$ denotes the parameterized problem $\SB(x,p(x)) \SM x\in L \SE$ where $x\in \Sigma^*$.  For a problem instance $(x,k) \in \Sigma^* \times \Nat$ we call $x$ the main part and $k$ the parameter.  
A parameterized problem $\PP$ is \emph{fixed-parameter   tractable} (FPT in short) if a given instance $(x, k)$ can be solved in time  $\bigoh(f(k) \cdot p(|x|))$ where $f$ is an arbitrary computable function of $k$ and $p$ is a polynomial function; we call algorithms running in this time \emph{fpt-algorithms}.

  Parameterized complexity classes are defined with respect to {\em fpt-reducibility}. A parameterized problem $P$ is {\em fpt-reducible} to $Q$ if in time $f(k)\cdot |x|^{O(1)}$, one can transform an instance $(x,k)$ of $\PP$ into an instance $(x',k')$ of $\QQ$ such that $(x,k)\in \PP$ if and only if $(x',k')\in \QQ$, and $k'\leq g(k)$, where $f$ and $g$ are computable functions depending only on $k$. 
 Owing to the definition, if $\PP$ fpt-reduces to $\QQ$
 and $\QQ$ is fixed-parameter tractable then $P$ is fixed-parameter
 tractable as well. 
 Central to parameterized complexity is the following hierarchy of complexity classes, defined by the closure of canonical problems under fpt-reductions:
 \[\FPT \subseteq \W{1} \subseteq \W{2} \subseteq \cdots \subseteq \XP.\] All inclusions are believed to be strict. In particular, $\FPT\neq \W{1}$ under the Exponential Time Hypothesis. 

 The class $\W{1}$ is the analog of $\NP$ in parameterized complexity. A major goal in parameterized complexity is to distinguish between parameterized problems which are in $\FPT$ and those which are $\W{1}$-hard, i.e., those to which every problem in $\W{1}$ is fpt-reducible. There are many problems shown to be complete for $\W{1}$, or equivalently $\W{1}$-complete, including the {\sc Multi-Colored Clique (MCC)} problem~\cite{DowneyFellows13}. We refer the reader to the respective monographs~\cite{FlumGrohe06,DowneyFellows13,CyganFKLMPPS15} for an in-depth
introduction to parameterized complexity.
 }
  \sv{
The parameterized complexity paradigm allows a finer analysis of the complexity of problems by associating each problem instance $L$ with a numerical parameter $k$; the pair $(L,k)$ is then an instance of a parameterized problem. A parameterized problem is \emph{fixed-parameter   tractable} (FPT in short) if a given instance $(L, k)$ can be solved in time  $\bigoh(f(k) \cdot |L|^{\bigoh(1)}$ where $f$ is an arbitrary computable function; we call algorithms running in this time \emph{fpt-algorithms}.  The complexity class $\W{1}$ is the analog of $\NP$ for parameterized complexity; under established complexity assumptions, problems that are hard for $\W{1}$ do not admit fpt-algorithms.

We refer the reader to the respective monographs~\cite{FlumGrohe06,DowneyFellows13,CyganFKLMPPS15} for an in-depth
introduction to parameterized complexity.
  }

\lv{\subsection{Integer Linear Programming}
Our algorithms use an Integer Linear Programming (ILP) subroutine. ILP is a well-known framework for formulating problems and a powerful tool for the development of {\FPT}-algorithms for optimization problems. 

\begin{definition}[$p$-Variable Integer Linear Programming Optimization] Let $A\in \mathbb{Z}^{q\times p}, b\in \mathbb{Z}^{q\times 1}$ and $c\in \mathbb{Z}^{1\times p}$. The task is to find a vector $x\in \mathbb{Z}^{p\times 1}$ which minimizes the objective function $c\times \bar x$ and satisfies all $q$ inequalities given by $A$ and $b$, specifically satisfies $A\cdot \bar x\geq b$. The number of variables $p$ is the parameter.
\end{definition}

Lenstra \cite{Lenstra83} showed that \textsc{$p$-ILP}, together with its optimization variant \textsc{$p$-OPT-ILP} (defined above), are in {\FPT}. His running time was subsequently improved by Kannan \cite{Kannan87} and Frank and Tardos \cite{FrankTardos87} (see also \cite{FellowsLokshtanovMisraRS08}).

\begin{theorem}[\cite{Lenstra83,Kannan87,FrankTardos87,FellowsLokshtanovMisraRS08}]
\label{thm:pilp}
\textsc{$p$-OPT-ILP} can be solved using $\bigoh(p^{2.5p+o(p)}\cdot L)$ arithmetic operations in space polynomial in $L$, $L$ being the number of bits in the input.
\end{theorem}
}

\subsection{Treewidth}\label{sec:tw}

A \emph{tree-decomposition} of a graph $G=(V,E)$ is a pair 
$(T,\{B_t : t\in V(T)\})$
where $B_t \subseteq V$ for every $t \in V(T)$ and $T$ is a tree such that:
\begin{enumerate}
  \item for each edge $(u,v)\in E$, there is a $t\in V(T)$ such that $\{u,v\} 
\subseteq B_t$, and \label{twone}
\item for each vertex $v\in V$, $T[\SB t\in V(T) \SM v\in B_t \SE]$ is a non-empty (connected) tree. \label{twtwo}
\end{enumerate}
The \emph{width} of a tree-decomposition is $\max_{t \in V(T)} |B_t|-1$.
The \emph{treewidth}~\cite{Kloks94} of $G$
is the minimum width taken over all tree-decompositions
of $G$ and it is denoted by $\tw(G)$. We call the elements of $V(T)$ \emph{nodes} and $B_t$ \emph{bags}.

While it is possible to compute the treewidth exactly using an fpt-algorithm~\cite{Bodlaender96}, the asymptotically best running time is achieved by using the recent state-of-the-art $5$-approximation algorithm of Bodlaender et al.~\cite{BodlaenderDDFLP16}. 

\begin{ourfact}[\cite{BodlaenderDDFLP16}]
\label{fact:findtw}
There exists an algorithm which, given an $n$-vertex graph $G$ and an integer~$k$, in time $2^{\bigoh(k)}\cdot n$ either outputs a tree-decomposition of $G$ of width at most $5k+4$ and $\bigoh(n)$ nodes, or correctly determines that $\tw(G)>k$.
\end{ourfact}

\lv{It is well known that, for every clique over
$Z\subseteq V(G)$ in $G$, it holds that every tree-decomposition of
$G$ contains an element $B_t$ such that $Z\subseteq B_t$
\cite{Kloks94}. 
Furthermore, if $t'$ separates a node~$t$ from another node $t''$ in $T$, then $B_{t'}$ separates $B_t\setminus B_{t'}$ from $B_{t''}\setminus B_{t'}$ in $G$~\cite{Kloks94}; this \emph{inseparability property} will be useful in some of our later proofs.
}
A tree-decomposition $(T,{B_t : t \in V(T)})$ of a graph~$G$ is
\emph{nice} if the following conditions hold:
\sv{(1) $T$ is rooted at a node $r$ such that $|B_r|=\emptyset$, (2) every
  node of $T$ has at most two children, if a node $t$ of $T$ has two children $t_1$ and $t_2$, then
  $B_t=B_{t_1}=B_{t_2}$; in that case we call $t$ a \emph{join
    node}, (3) if a node $t$ of $T$ has exactly one child $t'$, then exactly
  one of the following holds: (3A) $B_t=B_{t'}\cup \{v\}$, in which case
  we call $t$ an \emph{introduce node} or (3B) $B_{t}=B_{t'}\setminus
  \{v\}$ in which case we call $t$ a \emph{forget node}, and (4)
  if a node $t$ of $T$ is a leaf, then $|B_t|=1$; we call these \emph{leaf nodes}.
}

\lv{\begin{enumerate}
\item $T$ is rooted at a node $r$ such that $|B_r|=\emptyset$.
\item Every node of $T$ has at most two children.
\item If a node $t$ of $T$ has two children $t_1$ and $t_2$, then
  $B_t=B_{t_1}=B_{t_2}$; in that case we call $t$ a \emph{join
    node}.
\item If a node $t$ of $T$ has exactly one child $t'$, then exactly
  one of the following holds:
  \begin{enumerate}
  \item $\Card{B_t}=\Card{B_{t'}}+1$ and $B_{t'}\subset
    B_t$; in that case we call $t$ an \emph{introduce node}.
  \item $\Card{B_t}=\Card{B_{t'}}-1$ and $B_t\subset
    B_{t'}$; in that case we call $t$ a \emph{forget node}.
  \end{enumerate}
\item If a node $t$ of $T$ is a leaf, then $|B_t|=1$; we call these \emph{leaf nodes}.
\end{enumerate}}

The main advantage of nice tree-decompositions is that they allow the design of much more transparent dynamic programming algorithms, since one only needs to deal with four specific types of nodes. It is well known (and easy to see) that for every fixed~$k$, given a
tree-decomposition of a graph $G=(V,E)$ of width at most $k$ and with $\mathcal{O}(\Card{V})$
nodes, one can construct in linear time a nice
tree-decomposition of $G$ with $\mathcal{O}(\Card{V})$ nodes and width at
most~$k$~\cite{BodlaenderKloks96}. Given a node $t$ in $T$, we let $Y_t$ be the set of all vertices contained in the bags of the subtree rooted at $t$, i.e., $Y_t=B_t\cup \bigcup_{p \text{ is separated from the root by $t$}}B_p$.

\section{Closing the Gap on Graphs of Feedback Vertex Number One}

In this section we develop a polynomial-time algorithm for  EDP
restricted to graphs with feedback vertex set number one. We refer to
this particular variant as \textsc{Simple Edge Disjoint Paths (SEDP)}: 
given an EDP instance $(G,P)$ and a FVS $X=\{x\}$, solve $(G,P)$.
\lv{\pbDef{\textsc{Simple Edge Disjoint Paths (SEDP)}}
{A graph $G$, a minimal FVS $X=\{x\}\subseteq V(G)$ for $G$
  and a set $P$ of terminal pairs, i.e., a set of subsets of $V(G)$ of size two.}{Is there a
  set of pairwise edge disjoint paths connecting every set of two
  terminals in $P$?}
}Additionally to our standard assumptions about EDP (given in
Subsection~\ref{ssec:edp}), we will assume that: (1)
every neighbor of $x$ in $G$ is a leaf in $G - X$, (2)
$x$ is not a terminal, i.e., $x \notin
\bigunion{P}$, and (3) every tree $T$ in $G-X$ is rooted in a vertex
$r$ that is not a terminal. Property (1) can be ensured by an
additional leaf vertex $l$ to any non-leaf neighbor $n$ of $x$,
removing the edge $\{n,x\}$ and adding the edge $\{l,x\}$ to
$G$. Property (2) can be ensured by adding an additional leaf vertex $l$
to $x$ and replacing $x$ with $l$ in $P$ and finally (3) can be
ensured by adding a leaf vertex $l$ to any non-terminal vertex $r$ in
$T$ and replacing $r$ with $l$ in $P$.

A key observation behind our algorithm for SEDP is that whether or not
an instance $\III=(G,P,X)$ has a solution merely depends on the existence
of certain sets of pairwise edge disjoint paths in the trees $T$ in $G
- X$. In particular, as we will show in
Lemma~\ref{lem:dp-sol-cri} later on, $\III$ has a solution if and only if
every tree $T$ in $G -X$ is $\emptyset$-connected (see
Definition~\ref{def:connectedness}). The main ingredient of the
algorithm is then a bottom-up dynamic programming algorithm that
determines whether a tree $T$ in $G-X$ is $\emptyset$-connected. We now define the various connectivity states of subtrees of $T$ that
we need to keep track of in the dynamic programming table.
\begin{definition}\label{def:connectedness}
  Let $T$ be a tree in $G - X$ rooted at $r$ (recall that we can
  assume that $r$ is not in $\bigunion{P}$), $t \in V(T)$, and let $S$ be a set of
  pairwise edge disjoint paths in $G[T_t\cup X]$ and $P' \subseteq
  P(T_t)$, where $T_t$ is the subtree of $T$ rooted at $t$.

  We say that the set $S$ $\gamma_\emptyset$-\emph{connects} $P'$ in 
$G[T_t\cup X]$ if for every $a \in \bigunion{P'}\cap T_t$, the set $S$ either contains
an $a$-$x$ path disjoint from $b$, or it contains an $a$-$b$ 
path disjoint from $x$, where $\{a,b\}
\in P'$. Moreover, for $\ell \in \{\gamma_X\} \cup P(T_t)$,
we say that the set $S$ \emph{$\ell$-connects} $T_t$ if $S$ $\gamma_
\emptyset$-connects $P(T_t)
\setminus \{\ell\}$ and additionally the following conditions hold. 
\begin{itemize}
\item If $\ell=\gamma_X$ then $S$ also contains a path from
  $t$ to $x$.
\item If $\ell=p$ for some $p \in P(T_t)$ then:
  \begin{itemize}
  \item If $p\cap T_t=\{a\}$ then $S$ contains a $t$-$a$ path disjoint from  $x$.
  \item  If $p\cap T_t=\{a,b\}$ then $S$ contains a $t$-$a$ path disjoint from  $x$ and a $b$-$x$ path disjoint from $a$ or $S$ contains a $t$-$b$ path disjoint from  $x$ and an $a$-$x$ path disjoint from $b$. 
  \end{itemize}
\end{itemize}
For $\ell \in
\{\gamma_\emptyset,\gamma_X\} \cup P(T_t)$, we say that $T_t$ is $\ell$-connected if there is a set $S$ which $\ell$-connects $P(T_t)$ in $G[T_t\cup X]$.
\end{definition}
Informally, a tree $T_t$ is:
\sv{(1) $\gamma_\emptyset$-connected if all its terminal
  pairs can be connected in $G[T_t \cup X]$ either to themselves or to
  $x$, (2) $\gamma_X$-connected if it is
  $\gamma_\emptyset$-connected and additionally there is a path from its
  root to $x$ (which can later be used to connect some terminal not in
  $T_t$ to $x$ via the root of $T_t$), and (3)
  $\gamma_p$-connected if all but one of its terminals, i.e., one
  of the terminals in $p$, can be
  connected in $G[T_t \cup X]$ either to themselves or to
  $x$, and additionally one terminal in $p$ can be connected to the
  root of $T_t$ (from which it can later be connected to $x$ or the
  other terminal in $p$).}
\lv{\begin{itemize}
\item $\gamma_\emptyset$-connected if all its terminal
  pairs can be connected in $G[T_t \cup X]$ either to themselves or to
  $x$,
\item $\gamma_X$-connected if it is
  $\gamma_\emptyset$-connected and additionally there is a path from its
  root to $x$ (which can later be used to connect some terminal not in
  $T_t$ to $x$ via the root of $T$),
\item $\gamma_p$-connected if all but one of its terminals, i.e., one
  of the terminals in $p$, can be
  connected in $G[T_t \cup X]$ either to themselves or to
  $x$, and additionally one terminal in $p$ can be connected to the
  root of $T_t$ (from which it can later be connected to $x$ or the
  other terminal in $p$).
\end{itemize}
}
\lv{\begin{LEM}}
\sv{\begin{LEM}[$\star$]}
\label{lem:dp-sol-cri}
  $(G,X,P)$ has a solution if and only if every tree $T$ in $G-X$ is
  $\gamma_\emptyset$-connected.
\end{LEM}

\lv{\begin{proof}
  Let $S$ be a solution for $(G,X,P)$, i.e., a set of pairwise
  edge disjoint paths between all terminal pairs in $P$. Consider the
  set $S'$ of pairwise edge disjoint paths obtained from $S$ after
  splitting all paths in $S$ between two terminals $a$ and $b$ that
  intersect $x$, into two paths, one from $a$ to $x$ and the other from
  $x$ to $b$. Then the restriction of $S'$ to any tree $T$ in $G
  - X$ shows that $T$ is $\gamma_\emptyset$-connected.

  In the converse direction, for
  every tree $T$ in $G- X$, let $S_T$ be a set of pairwise
  edge disjoint paths witnessing that $T$ is $\gamma_\emptyset$-connected. Consider a set $p=\{a,b\} \in P$ and let $T_a$ and
  $T_b$ be the tree containing $a$ and  $b$, respectively. If $T_a=T_b$
  then either $S_{T_a}$ contains an $a$-$b$ path or $S_{T_a}$
  contains an $a$-$x$ path and a $b$-$x$ path. In both cases
  we obtain an $a$-$b$ path in $G$. Similarly, if $T_a\neq
  T_b$ then $S_{T_a}$ contains an $a$-$x$  path and $S_{T_b}$
  contains a $b$-$x$ path, whose concatenation gives us an $a$-$b$ path in $G$. Since the union of all $S_T$ over all trees $T$
  in $G - X$ is a set of pairwise edge disjoint paths, this
  shows that $(G,X,P)$ has a solution. This completes the proof of the lemma. 
\end{proof}
}
Due to Lemma \ref{lem:dp-sol-cri}, our algorithm to solve EDP only has
to determine whether every tree in $G- X$ is
$\gamma_\emptyset$-connected. For a tree $T$ in $G-X$, our algorithm achieves this by computing
a set of labels $L(t)$, where $L(t)$ is the set of all labels $\ell \in
\{\gamma_\emptyset,\gamma_X\} \cup P(T_t)$ such that $T_t$ is
$\ell$-connected, via a bottom-up dynamic programming procedure.
We begin
by arguing that for a leaf vertex $l$, the value $L(l)$ can be
computed in constant time. 

\begin{LEM}\label{lem:comp-labels-leaf}
  The set $L(l)$ for a leaf vertex $l$ of $T$ can be computed in time $\bigoh(1)$.
\end{LEM}
\begin{proof}
  Since $l$ is a leaf vertex, we conclude  that $T_l$ is
  $\gamma_\emptyset$-connected if and only if either $l \notin \bigunion{P}$
  or $l\in \bigunion{P}$ and $(l,x) \in E(G)$. Similarly,
  $T_l$ is $\gamma_X$-connected if and only if $l \notin
  \bigunion{P}$ and $(l,x) \in E(G)$. Finally, $T_l$ is
  $\ell$-connected for some $\ell \in P(T_l)$ if and only if $l \in
  \bigunion{P}$. Since all these properties can be checked in constant
  time, the statement of the lemma follows.
\end{proof}
We will next show how to compute $L(t)$ for a non-leaf vertex $t \in V(T)$ with
children $t_1,\dotsc,t_l$.
\begin{definition}
\sv{We define the three sets
$V_t^{\neg \gamma_\emptyset}=\SB t_i \SM \gamma_\emptyset
\notin L(t_i)\SE$, $V_t^{\gamma_X}=\SB t_i \SM \gamma_X \in
L(t_i)\SE$, and $V_t=\{t_1,\dotsc,t_l\}\setminus (V_t^{\neg \gamma_\emptyset} \cup
V_t^{\gamma_X})$.
}
\lv{We define the following three sets.  $$V_t^{\neg \gamma_\emptyset}=\SB t_i \SM \gamma_\emptyset
\notin L(t_i)\SE$$ $$V_t^{\gamma_X}=\SB t_i \SM
\gamma_X \in L(t_i)\SE$$  $$V_t=\{t_1,\dotsc,t_l\}\setminus (V_t^{\neg \gamma_\emptyset} \cup
V_t^{\gamma_X})$$}
\end{definition}
That is, $V_t^{\neg \gamma_\emptyset}$ is the set of those children $t_i$ such that $T_i$ is {\em not} $\gamma_\emptyset$-connected, $V_t^{\gamma_X}$ is the set of those children $t_i$ such that $T_i$ is $\gamma_X$-connected and $V_t$ is the set comprising the remaining children. 
 Observe that $\{V_t,V_t^{\neg
  \gamma_\emptyset},V_t^{\gamma_X}\}$ forms a partition of
$\{t_1,\dotsc,t_l\}$ and moreover $\gamma_\emptyset \in L(t)$ and
$\gamma_X \notin L(t)$ for every $t \in V_t$.
Let
$H(t)$ be the graph with vertex set $V_t \cup V_t^{\neg \gamma_\emptyset}$ having an
edge between $t_i$ and $t_j$ (for $i\neq j$) if and only if $L(t_i)
\cap L(t_j)\neq \emptyset$ and not both $t_i$ and $t_j$ are in $V_t$.
The following lemma is crucial to our algorithm, because it provides
us with a simple characterization of $L(t)$ for a non-leaf vertex $t
\in V(T)$.
\sv{\begin{LEM}[$\star$]}
\lv{\begin{LEM}}\label{lem:charinnervertex}
  Let $t$ be a non-leaf vertex of $T$ with children $t_1,\dotsc,t_l$.
  Then $T_t$ is:
  \begin{itemize}
  \item $\gamma_\emptyset$-connected if and only if $L(t') \neq \emptyset$
    for every $t' \in \{t_1,\dotsc,t_l\}$ and
    $H(t)$ has a matching $M$
    such that $|V^{\neg \gamma_\emptyset} \setminus V(M)|\leq |V_t^{\gamma_X}|$,
  \item $\gamma_X$-connected if and only if
    $L(t') \neq \emptyset$
    for every $t' \in \{t_1,\dotsc,t_l\}$ and
    $H(t)$ has a matching $M$
    such that $|V^{\neg \gamma_\emptyset} \setminus V(M)|< |V_t^{\gamma_X}|$,
  \item $\ell$-connected (for $\ell \in P(T_t)$) if and only if
    $L(t') \neq \emptyset$
    for every $t' \in \{t_1,\dotsc,t_l\}$ and
    there
    is a $t_i$ with $\ell \in L(t_i)$ such that $H(t) -
    \{t_i\}$ has a matching $M$ with $|V^{\neg \gamma_\emptyset} \setminus V(M)|\leq |V_t^{\gamma_X}|$.
  \end{itemize}
\end{LEM}
\lv{\begin{proof}
  Towards showing the forward direction let $S$ be a set of
  pairwise edge disjoint paths witnessing that $T_t$ is
  $\ell$-connected for some $\ell \in \{\gamma_\emptyset,\gamma_X\}\cup
  P(T_t)$. Then $S$ contains the following types of paths:
  \begin{description}
  \item[(T1)] A path between $a$ and $b$ that does not contain
    $x$, where $\{a,b\} \in P(T_t)$,
  \item[(T2)] A path between $a$ and $x$, which does
    not contain $b$, where $\{a,b\} \in P(T_t)$,
  \item[(T3)] A path between $x$ and $t$ (only if $\ell=\gamma_X$),
  \item[(T4)] A path between $a$ and $t$, which does
    not contain $x$, (only if $\ell=p$ for some $p \in P(T_t)$ and $a \in p$),   
  \end{description}
  For every $i$ with $1 \leq i \leq l$, let $S_i$ be the subset of $S$
  containing all paths that use at least one vertex in $T_{t_i}$ and
  let $S_i'$ be the restriction of all paths in $S_i$ to
  $G[T_{t_i} \cup X]$. Consider an $i$ with $1 \leq i \leq
  l$. Because the paths in $S$ are pairwise edge disjoint, we obtain
  that at most one path in $S$ contains the edge $(t_i,t)$.
  We start with the following observations.
  \begin{itemize}
  \item[(O1)] If $S_i$ contains no path that contains the edge $(t_i,t)$,
    then $S_i'$ shows that $T_{t_i}$ is $\gamma_\emptyset$-connected.
  \item[(O2)] If $S_i$ contains a path say $s_i$ that contains the edge $(t_i,t)$,
    then the following statements hold.
    \begin{itemize}
    \item[(O21)] If $s_i$ is a path of Type (T1) for some $p
      \in P(T_{t})$, then $S_i'$ shows that $T_{t_i}$ is
      $p$-connected,
    \item[(O22)] If $s_i$ is a path of Type (T2) for some $p
      \in P(T_t)$, then the following statements hold. 
      \begin{itemize}
      \item[(O221)] If the endpoint of $s_i$ in $G[T_{t_i}\cup X]$ is a
        terminal $t \in p$, then $S_i'$ shows that $T_{t_i}$ is
        $p$-connected.
      \item[(O222)] Otherwise, i.e., if the endpoint of $s_i$ in
        $G[T_{t_i}\cup X]$ is $x$, then $S_i'$ shows that $T_{t_i}$ is
        $\gamma_X$-connected.
      \end{itemize}
         \end{itemize}
    \item[(O3)] If $s_i$ is a path of Type (T3), then $S_i'$ shows that
      $T_{t_i}$ is $\gamma_X$-connected.
    \item[(O4)] If $s_i$ is a path of Type (T4), for some terminal-pair $p
      \in P(T_{t_i})$ then $S_i'$ shows that $T_{t_i}$ is $p$-connected.
 
  \end{itemize}
  Let $M$ be the set of all pairs $\{t_i,t_j\}\subseteq V_t \cup
  V_t^{\neg \gamma_\emptyset}$ such that $t_i \in V_t^{\neg \gamma_\emptyset}$ or
  $t_j \in V_t^{\neg \gamma_\emptyset}$ and $S$ contains a path that contains
  both edges $(t_i,t)$ and $(t,t_j)$. We claim that $M$ is a
  matching of $H(t)$. Since an edge $(t_i,t)$ can be used by at most
  one path in $S$, it follows that the pairs in $M$ are pairwise
  disjoint and it remains to show that $M \subseteq E(H(t))$, namely,
  $L(t_i) \cap L(t_j)\neq \emptyset$ for every $(t_i,t_j) \in M$.
  Let $s \in S$ be the path witnessing that $(t_i,t_j) \in M$.
  Because $s$ contains the edges $(t_i,t)$ and $(t,t_j)$, it
  cannot be of Type (T3) or (T4). Moreover, if $s$ is of Type (T2)
  then either $T_{t_i}$ or $T_{t_j}$ is $\gamma_X$-connected,
  contradicting our assumption that $t_i,t_j \in V_t \cup V_t^{\neg
    \gamma_\emptyset}$. Hence $s$ is of Type (T1) for some terminal-pair $p
  \in P(T_t)$, which implies that $T_{t_i}$ and $T_{t_j}$ are
  $p$-connected, as required.

  In the following let $t_i \in V^{\neg \gamma_\emptyset}$.
  Because of Observation (O1), we obtain that $S_i$ contains a path,
  say $s_i$, using the edge $(t_i,t)$ (otherwise $T_{t_i}$ is
  $\gamma_\emptyset$-connected). Moreover, together with Observations (O2)--(O4), we
  obtain that either:
  \begin{description}
  \item[(P1)] $s_i$ is a path of Type (T1) for some $p \in P(T_{t})$,
  \item[(P2)] $s_i$ is a path of Type (T2) for some $p \in P(T_t)$ and
    the endpoint of $s_i$ in $G[T_{t_i}\cup X]$ is a terminal $t \in
    p$,
  \item[(P3)] $s_i$ is a path of Type (T4) for some $p \in P(T_{t})$.
  \end{description}
  Because $t$ is not connected to $x$ and $t$ is not a terminal, we
  obtain that if $s_i$ satisfies (P1), then there is a $t_j \in
  \{t_1,\dotsc,t_j\}$ with $p \in L(t_j)$ such that $s_i$ contains the
  edge $(t_j,t)$. Similarly, if $s_i$ satisfies (P2) then there is
  a $t_j \in V_t^{\gamma_X}$ such that $s_i$ contains the edge
  $(t_j,t)$. Consequently, every $t_i \in V_t^{\neg \gamma_\emptyset}$ for
  which $s_i$ satisfies (P1) or (P2) is
  mapped to a unique $t_j \in \{t_1,\dotsc,t_l\}$ such that $s_i$
  contains the edge $(t_j,t)$.
  
  We now distinguish three cases depending on $\ell$.
  If $\ell=\emptyset$, then $S$ contains no path of type (T4) and
  hence for every $t_i \in V_t^{\neg \gamma_\emptyset}$ either (P1) or (P2)
  has to hold. In particular, for every $t_i \in V_t^{\neg
    \gamma_\emptyset}\setminus V(M)$
  there must be a $t_j \in V_t^{\gamma_X}$ such that $s_i$ contains
  the edge $(t_j,t)$. Consequently, if $|V_t^{\gamma_X}|<|V^{\neg
    \gamma_\emptyset}\setminus V(M)|$, this is not achievable
  contradicting our assumption that $S$ is a witness to the fact that
  $T_t$ is $\gamma_\emptyset$-connected.

  If $\ell=\gamma_X$, then $S$ contains a path of Type
  (T3), which due to Observation (O3) uses the edge $(t',t)$ for
  some $t' \in V_t^{\gamma_X}$. Since again (P1) or (P2) has to
  hold for every $t_i \in V_t^{\neg \gamma_\emptyset}$, we obtain that
  for every $t_i \in V_t^{\neg \gamma_\emptyset}\setminus V(M)$ there
  must be a $t_j \in V_t^{\gamma_X}\setminus \{t'\}$ such that $s_i$ contains the
  edge $(t_j,t)$. Consequently, if $|V_t^{\gamma_X}|\leq|V^{\neg
    \emptyset}\setminus V(M)$ then $|V_t^{\gamma_X}\setminus\{t_j\}|<|V^{\neg
    \emptyset}\setminus V(M)$ and this is not achievable
  contradicting our assumption that $S$ is a witness to the fact that
  $T_t$ is $\gamma_X$-connected.

  Finally, if $\ell=p$ for some $p \in P(T_t)$, then $S$ contains a
  path of Type (T4), which due to Observation (O4) uses the edge
  $(t',t)$ for some $t' \in \{t_1,\dotsc,t_l\}$ with $p \in
  L(t')$. Observe that $t'$ cannot be matched by $M$ and hence $M$ is
  also a matching of $H(t) - \{t'\}$. Moreover,
  Since (P1) or (P2) has to
  hold for every $t_i \in V_t^{\neg \gamma_\emptyset}\setminus \{t_j\}$, we obtain that
  for every $t_i \in V_t^{\neg \gamma_\emptyset}\setminus
  (V(M)\setminus \{t_j\}$ there
  must be a $t_j \in V_t^{\gamma_X}\setminus \{t'\}$ such that $s_i$ contains the
  edge $(t_j,t)$. Consequently, if
  $|V_t^{\gamma_X}\setminus \{t'\}|< |V^{\neg
    \emptyset}\setminus(V(M) \cup \{t'\})|$
  then this is not achievable
  contradicting our assumption that $S$ is a witness to the fact that
  $T_t$ is $p$-connected.
  
  Towards showing the reverse direction we will show how to construct
  a set $S$ of pairwise edge disjoint paths witnessing that $T_t$ is 
  $\ell$-connected. 

  If $\ell=\emptyset$, then let $M$ be a matching of
  $H(t)$ such that $|V^{\neg \gamma_\emptyset}\setminus V(M)|\leq
  |V_t^{\gamma_X}|$ and let $\alpha$ be a bijection between
  $V^{\neg \gamma_\emptyset}\setminus V(M)$ and $|V^{\neg
    \emptyset}\setminus V(M)|$ arbitrarily chosen elements in
  $V_t^{\gamma_X}$. Then the set $S$ of pairwise edge disjoint
  paths witnessing that $T_t$ is $\gamma_\emptyset$-connected is obtained as follows.
  \begin{description}
  \item[(S1)] For every $t_i \in V_t \setminus V(M)$, let $S'$ be a
    set of pairwise edge disjoint paths witnessing that $T_{t_i}$ is
    $\gamma_\emptyset$-connected. Add all paths in $S'$ to $S$.
  \item[(S2)] For every $t_i \in V_t^{\gamma_X} \setminus \SB \alpha(t')
    \SM t' \in V^{\neg \gamma_\emptyset}\setminus V(M)\SE$, let $S'$ be a
    set of pairwise edge disjoint paths witnessing that $T_{t_i}$ is
    $\gamma_\emptyset$-connected. Add all paths in $S'$ to $S$.
  \item[(S3)] For every $(t_i,t_j) \in M$, choose $p \in L(t_i) \cap
    P(T_j)$ arbitrarily and let $S_1$ and $S_2$ be sets of
    pairwise edge disjoint paths witnessing that $T_{t_i}$ and
    $T_{t_j}$ are $p$-connected, respectively.
    Moreover, let $s_1 \in S_1$ and $s_2 \in S_2$ be the unique paths
    connecting a terminal in $p$ with $t_i$ and $t_j$, respectively.
    Then add all path in $(S_1
    \setminus \{s_1\}) \cup (S_2\setminus \{s_2\})$ and additionally
    the path $s_1\concat (t_i,t,t_j)\concat s_2$ to $S$.
  \item[(S4)] For every $t_i \in V_t^{\neg \gamma_\emptyset} \setminus
    V(M)$, choose $p \in L(t_i)$ arbitrarily and let $S_1$ be a
    set of pairwise edge disjoint paths witnessing that $T_{t_i}$ is
    $p$-connected. Moreover, let $s_1$ be the unique path in $S_1$ connecting
    a terminal in $p$ with $t_i$. Furthermore, let $S_2$ be a set of
    pairwise edge disjoint paths witnessing that $T_{\alpha(t_i)}$ is
    $\gamma_X$-connected and let $s_2$ be the unique path in $S_2$
    connecting $x$ with $\alpha(t_i)$. Then add all path in $(S_1
    \setminus \{s_1\}) \cup (S_2\setminus \{s_2\})$ and additionally
    the path $s_1\concat (t_i,t,\alpha(t_i))\concat s_2$ to $S$.
  \end{description}
  It is straightforward to verify that $S$ witnesses that $T_t$ is
  $\gamma_\emptyset$-connected.

  If $\ell=\gamma_X$, then the set $S$ of pairwise edge disjoint
  paths witnessing that $T_t$ is $\gamma_X$-connected is defined
  analogously to the case where $\ell=\emptyset$ only that now step
  (S2) is replaced by:
  \begin{description}
  \item[(S2')] Chose one $t' \in V_t^{\gamma_X} \setminus \SB \alpha(t')
    \SM t' \in V^{\neg \gamma_\emptyset}\setminus V(M)\SE$
    arbitrarily; note that such a $t'$ must exist because $|V_t^{\neg
      \emptyset}\setminus V(M)|< |V_t^{\gamma_X}|$.
    Let $S'$ be a set of pairwise edge disjoint paths witnessing that
    $T_{t'}$ is $\gamma_X$-connected and let $s'$ be the unique
    path in $S'$ connecting $t'$ with $x$. Then add all paths in
    $S'\setminus s'$ to $S$ and additionally the path obtained from
    $s'$ after adding the edge $(t',t)$. Finally, for every child
    in $V_t^{\gamma_X} \setminus (\SB \alpha(t')
    \SM t' \in V^{\neg \gamma_\emptyset}\setminus V(M)\SE \cup
    \{t'\})$ proceed as described in Step (S2).
  \end{description}

  Finally, if $\ell=p$ for some $p \in P(T_t)$, then let $t_i$ be a
  child with $p \in L(t_i)$ and let $M$ be a matching of
  $H(t)- \{t_i\}$ such that $|V^{\neg \gamma_\emptyset}\setminus V(M)|\leq
  |V_t^{\gamma_X}|$. Moreover, and let $\alpha$ be a bijection between
  $V^{\neg \gamma_\emptyset}\setminus (V(M) \cup \{t_i\})$ and $|V^{\neg
    \emptyset}\setminus (V(M) \cup \{t_i\})|$ arbitrarily chosen elements in
  $V_t^{\gamma_X}\setminus \{t_i\}$. Then the set $S$ of pairwise edge disjoint
  paths witnessing that $T_t$ is $p$-connected is obtained analogously
  to the case where $\ell=\emptyset$ above only that we do not execute
  the steps (S1)--(S4) for $t_i$ but instead do the following:
  \begin{description}
  \item[(S0)] Let $S'$ be a set of pairwise edge disjoint paths
    witnessing that $T_{t_i}$ is $p$-connected and let $s'$ be the
    unique path in $S'$ connecting a terminal in $p$ to $t_i$. Then we
    add all path in $S' \setminus \{s'\}$ and additionally the path
    $s'\concat(t_i,t)$ to $S$.
  \end{description}
  This completes the proof of Lemma \ref{lem:charinnervertex}. 
\end{proof}
}
The following two lemmas show how the above characterization can be
employed to compute $L(t)$ for a non-leaf vertex $t$ of $T$. Since the
matching employed in Lemma~\ref{lem:charinnervertex} needs to maximize
the number of vertices covered in $V^{\neg \gamma_\emptyset}$, we
first show how such a matching can be computed efficiently.
\sv{\begin{LEM}[$\star$]}
\lv{\begin{LEM}}\label{lem:max-match}
  There is an algorithm that, given a graph
  $G$ and a subset $S$ of $V(G)$, computes a matching $M$ maximizing
  $|V(M) \cap S|$ in time $\bigoh(\sqrt{|V|}|E|)$.
\end{LEM}
\lv{\begin{proof}
  We reduce the problem to \textsc{Maximum Weighted Matching} problem,
  which can be solved in time
  $\bigoh(\sqrt{|V|}|E|)$~\cite{MicaliVazirani80}. The reduction simply
  assigns weight two to every edge that is completely contained in $S$,
  weight one to every edge between $S$ and $V(G) \setminus S$, and
  weight zero otherwise. The correctness of the reduction follows
  because the weight of any matching $M$ in the
  weighted graph equals the number of vertices in $S$ covered by $M$.
\end{proof}
}
\lv{\begin{LEM}}
\sv{\begin{LEM}[$\star$]}
\label{lem:comp-labels-inner}
  Let $t$ be a non-leaf vertex of $T$ with children
  $t_1,\dotsc,t_l$. Then $L(t)$ can be computed from
  $L(t_1),\dotsc,L(t_l)$ in time $\bigoh(|P(T_t)|l^2\sqrt{l})$.
\end{LEM}
\lv{\begin{proof}
  First we construct the graph $H(t)$ in time $\bigoh(l^2|P(T_t)|)$. We
  then check for every $\ell \in \{\gamma_\emptyset,\gamma_X\}\cup
  P(T_t)$, whether $T_t$ is $\ell$-connected with the help of
  Lemma~\ref{lem:charinnervertex} as follows. If $\ell \in
  \{\gamma_\emptyset,\gamma_X\}$ we compute a matching $M$ in $H(t)$ that
  maximizes $|V(M) \cap V_t^{\neg \gamma_\emptyset}|$, i.e., the
  number of matched vertices in $V_t^{\neg \gamma_\emptyset}$. This can be
  achieved according to Lemma~\ref{lem:max-match} in time
  $\bigoh(\sqrt{l}l^2)$. Then in
  accordance with Lemma~\ref{lem:charinnervertex}, we add $\emptyset$
  or $\gamma_X$ to $L(t)$ if
  $|V^{\neg \gamma_\emptyset} \setminus V(M)|\leq
  |V_t^{\gamma_X}|$ or $|V^{\neg \gamma_\emptyset} \setminus
  V(M)|\leq |V_t^{\gamma_X}|$, respectively.
  For any $\ell \in P(T_t)$, again in accordance with
  Lemma~\ref{lem:charinnervertex}, we compute for every child $t'$ of $t$
  with $\ell \in P(T_{t'})$, a matching $M$ in $H(t) - \{t'\}$
  that maximizes $|V(M) \cap V_t^{\neg \gamma_\emptyset}|$. Since $t$
  has at most two children $t'$ with $\ell \in P(T_{t'})$ and
  due to Lemma~\ref{lem:max-match} such a matching can be
  computed in time $\bigoh(\sqrt{l}l^2)$, this is also the total running time for this
  step of the algorithm. If one of the at most two such matchings $M$
  satisfies $|V^{\neg \gamma_\emptyset} \setminus V(M)|\leq
  |V_t^{\gamma_X}|$, we add $\ell$ to $L(t)$ and otherwise we do
  not.
  This completes the description of the algorithm whose total running
  time can be obtained as the time required to construct
  $H(t)$ ($\bigoh(l^2|P(T_t)|$) plus $|P(T_t)|+2$ times the time required
  to compute the required matching in $H(t)$ ($\bigoh(\sqrt{l}l^2)$), which
  results in a total running time of
  $\bigoh(l^2|P(T_t)|+|P(T_t)|\sqrt{l}l^2)=O(|P(T_t)|l^2\sqrt{l})$.
\end{proof}}
We are now ready to put everything together into an algorithm that
decides whether a tree $T$ is $\gamma_\emptyset$-connected.
\begin{LEM}\label{lem:dp-alg}
  Let $T$ be a tree in $G - X$. There is an algorithm that
  decides whether $T$ is $\gamma_\emptyset$-connected in time $\bigoh(|P(T)||V(T)|^{\frac{5}{2}})$.
\end{LEM}
\begin{proof}
  The algorithm computes the set of labels $L(t)$ for every vertex $t
  \in V(T)$ using a bottom-up dynamic programming approach. Starting from
  the leaves of $T$, for which the set of labels can be computed due to
  Lemma~\ref{lem:comp-labels-leaf} in constant time, it uses
  Lemma~\ref{lem:comp-labels-inner} to compute $L(t)$ for every inner
  node $t$ of $T$ in time $\bigoh(|P(T_t)|l^2\sqrt{l})$. The total running
  time of the algorithm is then the sum of the running time for any
  inner node of $T$ plus the number of leaves of $T$, i.e.,
  $\bigoh(|P(T)||V(T)|^{\frac{5}{2}})$.
\end{proof}
\setcounter{THE}{0}
\begin{THE}
  {\sc SEDP} can be solved in time 
  $\bigoh(|P||V(G)|^{\frac{5}{2}})$.
\end{THE}
\begin{proof}
  We first employ Lemma~\ref{lem:dp-alg} to determine whether every
  tree $T$ of $G - X$ is $\gamma_\emptyset$-connected. If so we
  output \textsc{Yes} and otherwise \textsc{No}. Correctness
  follows from Lemma~\ref{lem:dp-sol-cri}.
\end{proof}

\section{Treewidth and Maximum Degree}

The goal of this section is to obtain an fpt-algorithm for
EDP parameterized by the treewidth $\omega$ and maximum degree $\Delta$ of the input graph.
\setcounter{THE}{1}
\begin{THE}
  \textsc{EDP} can be solved in time $2^{\bigoh(\Delta\omega^2)}\cdot n$, where $\omega$, $\Delta$ and $n$ are the treewidth, maximum degree and number of vertices of the input graph $G$, respectively.
\end{THE}

\begin{proof}
Let $(G,P)$ be an instance of \textsc{EDP} and let $(T,\cal B)$ be a
nice tree-decomposition of $G$ of width at most $k=5\omega+4$; recall
that such $(T,\cal B)$ can be computed in time $2^{\bigoh(k)}$ by Fact~\ref{fact:findtw}.
Consider the following leaf-to-root dynamic programming algorithm $\AA$, executed on $T$. 
At each bag $B_t$ associated with a node $t$ of $T$, $\AA$ will compute a table ${\cal M}_t$ of \emph{records}, which are tuples of the form $\{(\textsf{used},\textsf{give},\textsf{single})\}$ where:
\begin{itemize}
\item {\sf used} is a multiset of subsets of $B_t$ of cardinality $2$ with
  each subset occurring at most $\Delta$ times,
\item {\sf give} is a mapping from subsets of $B_t$ of cardinality $2$ to $[\Delta]$, and
\item {\sf single} is a mapping which maps each terminal $a_i\in Y_t$ such that its counterpart $b_i\not \in Y_t$ to an element of $B_t$.
\end{itemize}

Before we proceed to describe the steps of the algorithm itself, let
us first introduce the semantics of a record. For a fixed $t$, we will
consider the graph $G_t$ obtained from $G[Y_t]$ by removing all edges
with both endpoints in $B_t$ (we note that this ``pruned'' definition
of $G_t$ is not strictly necessary for the algorithm, but makes
certain steps easier later on). Then
$\alpha=\{(\textsf{used},\textsf{give},\textsf{single})\}\in {\cal
  M}_t$ if and only if there exists a set of edge disjoint paths $Q$
in $G_t$ and a surjective mapping $\tau$ from terminal pairs occurring
in $Y_t$ to subsets of $B_t$ of size two with the following properties:
\begin{itemize}
\item For each terminal pair $ab$ that occurs in $Y_t$: 
\begin{itemize}
\item[$\bullet$] $Q$ either contains a path whose endpoints are $a$ and $b$, or 
\item[$\bullet$] $Q$ contains an $a$-$x_1$ path for some $x_1\in B_t$ and a $b$-$x_2$ path for some $x_2\in B_t$ which is distinct from $x_1$, and furthermore $\tau(ab)=\{x_1,x_2\}\in$used;
\end{itemize}
\item for each terminal pair $ab$ such that $a\in Y_t$ but $b\not \in Y_t$:
\begin{itemize}
\item[$\bullet$] $Q$ contains a path whose endpoints are $a$ and $x\in B_t$, where $(a,x)\in \textsf{single}$;
\end{itemize}
\item for each distinct $x_1,x_2\in B_t$, $Q$ contains precisely $\textsf{give}(\{x_1,x_2\})$ paths from $x_1$ to $x_2$.
\end{itemize}

In the above case, we say that $Q$ witnesses $\alpha$. It is important
to note that the equivalence between the existence of records and sets
$Q$ of pairwise edge disjoint paths only holds because of the bound on the
maximum degree. That is because every vertex of $G$ has degree at most
$\Delta$, it follows that any set $Q$ of pairwise edge
disjoint paths can contain at most $\Delta$ paths containing a vertex
in the boundary.
Moreover, we note that by
reversing the above considerations, given a set of edge disjoint paths
$Q$ in $G_t$ satisfying a certain set of conditions, we can construct
in time $3^{\Delta k}$ a set of records in ${\cal M}_t$ that are
witnessed by $Q$ (one merely needs to branch over all options of
assigning paths in $\alpha$ which end in the boundary: they may either
contribute to give or to single or to used). These conditions are that
each path either (i) connects a terminal pair, (ii) connects a terminal
pair to two vertices in $B_t$, (iii) connects two vertices in
$B_t$, or (iv) connects a terminal $a \in Y_t$ whose counterpart $b \notin Y_t$ to a vertex in $B_t$.

$\AA$ runs as follows: it begins by computing the records ${\cal M}_t$ for each leaf $t$ of $T$. It then proceeds to compute the records for all remaining nodes in $T$ in a bottom-up fashion, until it computes ${\cal M}_r$. Since $B_r=\emptyset$, it follows that $(G,P)$ is a yes-instance if and only if $(\emptyset, \emptyset, \emptyset)\in {\cal M}_r$. For each record $\alpha$, it will keep (for convenience) a set $Q_\alpha$ of edge disjoint paths witnessing $\alpha$. Observe that while for each specific $\alpha$ there may exist many possible choices of $Q_\alpha$, all of these interact with $B_t$ in the same way.

We make one last digression before giving the procedures used to
compute ${\cal M}_t$ for the four types of nodes in nice
tree-decompositions. First, observe that the size of one particular
record is at most $\Delta k^2+\Delta k^2+|\textsf{single}|$. Since the
number of edge disjoint paths in $G_t$ ending in $B_t$ is
upper-bounded by $\Delta k$, it follows that each record in ${\cal
  M}_t$ satisfies $|\textsf{single}|\leq \Delta k$ and in particular
each such record has size at most $\bigoh(\Delta k^2)$. As a
consequence, $|{\cal M}_t|\leq 2^{\bigoh(\Delta k^2)}$ for each node
$t$ in $T$, which is crucial to obtain an upper-bound on the running time of
$\AA$.
\sv{Due to space reasons the details on how the records can be
  computed for the four types of nodes of a nice tree decomposition, i.e., leaf, introduce, join,
  and forget, have been moved to the appendix.
  }
\lv{
\begin{description}[leftmargin = 0cm]
\item[Case 1:]{\em $t$ is a leaf node.} 
If $v\in B_t$ is not a terminal, then ${\cal M}_t=\{\emptyset, \emptyset, \emptyset\}$. On the other hand, if $v\in B_t$ is a terminal, then ${\cal M}_t=\{\emptyset, \emptyset, \{(v,v)\}\}$.

\item[Case 2:]{\em $t$ is an introduce node:}
Let $p$ be the child of $t$ in $T$ and let $v$ be the vertex introduced at $t$, i.e., $v\in B_t\setminus B_p$. Recall that $v$ has no neighbors in $G_t$. As a consequence, if $v$ is not a terminal, then ${\cal M}_t={\cal M}_p$. If $v$ is a terminal and its counterpart $w$ lies in $G_t$, then for each $\{(\textsf{used},\textsf{give},\textsf{single})\}\in {\cal M}_p$ we add the record $\{(\textsf{used}',\textsf{give}',\textsf{single}')\}$ to ${\cal M}_t$, where:
\begin{itemize}
\item $\textsf{give}'=\textsf{give}$,
\item $\textsf{single}'$ is obtained by deleting the unique tuple $(w,?)$ from single, and
\item $\textsf{used}'$ is obtained by adding the subset $\{v, ?\}$ to used.
\end{itemize}
Finally, if $v$ is a terminal and its counterpart $w$ does not lie in $G_t$, then for each $\{(\textsf{used},\textsf{give},\textsf{single})\}\in {\cal M}_p$ we add the record $\{(\textsf{used},\textsf{give},\textsf{single}\cup{(v,v)})\}$ to ${\cal M}_t$.

\item[Case 3:] {\em $t$ is a join node:}
Let $p$, $q$ be the two children of $t$ in $T$. For each $\{(\textsf{used}_p,\textsf{give}_p,\textsf{single}_p)\}\in {\cal M}_p$ and each $\{(\textsf{used}_q,\textsf{give}_q,\textsf{single}_q)\}\in {\cal M}_p$, we add a new record $\{(\textsf{used},\textsf{give},\textsf{single})\}$ to ${\cal M}_t$ constructed as follows:
\begin{enumerate}
\item for each $x,y\in B_t$, we set $\textsf{give}(xy):=\textsf{give}_p(xy)+\textsf{give}_q(xy)$;
\item for each terminal pair $v,w$ such that $(v,?_v)\in \textsf{single}_p$ and $(w,?_w)\in \textsf{single}_q$ where $?_v=?_w$, 
\begin{itemize}
\item[$\bullet$] we delete $(v,?_v)$ from $\textsf{single}_p$ and 
\item[$\bullet$] we delete $(w,?_w)$ from $\textsf{single}_q$;
\end{itemize}
\item for each terminal pair $v,w$ such that $(v,?_v)\in \textsf{single}_p$ and $(w,?_w)\in \textsf{single}_q$ where $?_v\neq ?_w$, 
\begin{itemize}
\item[$\bullet$] we delete $(v,?_v)$ from $\textsf{single}_p$ and 
\item[$\bullet$] we delete $(w,?_w)$ from $\textsf{single}_q$ and
\item[$\bullet$] we add $\{?_v,?_w\}$ to used;
\end{itemize}
\item we set $\textsf{used}:=\textsf{used}\cup \textsf{used}_p \cup \textsf{used}_q$;
\item we set $\textsf{single}:=\textsf{single}_p \cup \textsf{single}_q$;
\item finally, we restore the records in ${\cal M}_p$ and ${\cal M}_q$ to their original state as of step $1$ (i.e., we restore all deleted items).
\end{enumerate}

\item[Case 4:]{\em $t$ is a forget node:}
Let $p$ by the child of $t$ in $T$ and let $v$ be the vertex forgotten at $t$, i.e., $v\in B_p\setminus B_t$. We note that this will be the by far most complicated step, since we are forced to account for the edges between $v$ and its neighbors in $B_t$. 

Let us begin by considering how the records in ${\cal M}_t$ can be obtained from those in ${\cal M}_p$; in particular, let us fix an arbitrary $\alpha=(\textsf{used}_t,\textsf{give}_t,\textsf{single}_t)\in {\cal M}_t$. This means that there exists a set $Q_\alpha$  of edge disjoint paths (along with a mapping $\tau$) in $G_t$ satisfying the conditions given by $\alpha$. Furthermore, let $E_v=\SB wv\in E(G)\SM w\in B_t\SE$, i.e., $E_v$ is the set of edges which are not present in $G_p$ but are present in $G_t$.
The crucial observation is that for each set $Q_\alpha$, there exists a set $Q_\beta$ of edge disjoint paths satisfying the conditions given by some $\beta=(\textsf{used}_p,\textsf{give}_p,\textsf{single}_p)\in {\cal M}_p$ such that $Q_\alpha$ can be obtained by merely ``extending'' $Q_\beta$ using the edges in $E_v$; in particular, $E_v$ can increase the number of paths contributing to $\textsf{give}_p$, change the endpoints of paths captured by $\textsf{used}_p$, and also change the endpoints of paths captured by $\textsf{single}_p$.

Formally, we will proceed as follows. For each $\beta\in {\cal M}_p$,
we begin with the witness $Q_\beta$ stored by $\AA$, and then branch
over all options of how the addition of the edges in $E_v$ may be used
to augment the paths in $Q_\beta$. In particular, since $|E_v|\leq k$,
there are at most $k+1$ vertices incident to an edge in $E_v$, and due
to the degree bound of $\Delta$ it then follows that there are at most
$\Delta(k+1)$ distinct paths in $Q_\beta$ which may be augmented by
$E_v$. Since each edge in $E_v$ may either be assigned to extend one
such path or form a new path, we conclude that there exist at most
$k^{\bigoh(\Delta k)}$ possible resulting sets of edge-disjoint paths
in $G_t$. For each such set $Q'$ of edge-disjoint paths, we then check
that conditions (i)-(iv) for constructing records witnessed by $Q'$
hold, and in the positive case we construct these records in time
$3^{\Delta k}$ as argued earlier and add them to ${\cal M}_t$.
\end{description}
}

\paragraph*{Summary and running time.}
Algorithm $\AA$ begins by invoking Fact~\ref{fact:findtw} to compute a
tree-decomposition of width at most $k=5\omega+4$ and $\bigoh(n)$
nodes, and then converts this into a nice tree-decomposition $(T,\cal
B)$ of the same width and also $\bigoh(n)$ nodes. It then proceeds to
compute the records ${\cal M}_t$ (along with corresponding witnesses)
for each node $t$ of $T$ in a leaves-to-root fashion, using the
procedures described above. The number of times any procedure is
called is upper-bounded by $\bigoh(n)$, and the running time of every
procedure is upper-bounded by the worst-case running time of the
procedure for forget nodes. There, for each record $\beta$ in the data
table of the child of $t$, the algorithm takes its witness $Q_\beta$
and uses branching to construct at most $k^{\bigoh(\Delta k)}$ new
witnesses (after the necessary conditions are checked). Each such
witness $Q_\alpha$ gives rise to a set of records that can be computed
in time $3^{\Delta k}$, which are then added to ${\cal M}_t$ (unless
they are already there). All in all, the running time of this
procedure is upper-bounded by $2^{\bigoh(\Delta k^2)}\cdot
k^{\bigoh(\Delta k)}\cdot 3^{\Delta k}=2^{\bigoh(\Delta k^2)}$, and the run-time of the algorithm follows.
\end{proof}

\section{Lower Bounds of EDP for Parameters of the Augmented Graph}

\sv{
  This section is devoted to a proof of the following theorem.
  \begin{THE}[$\star$]
  \label{the:hard-edp-fvs}
    EDP is \W{1}-hard parameterized by the feedback vertex set number
    of the augmented graph.
  \end{THE}
  Note that since the feedback vertex set number upper-bounds the
  treewidth, Theorem~\ref{the:hard-edp-fvs} complements the
  result in~\cite{ZhouTN00} showing that EDP is solvable in
  polynomial time for bounded treewidth. 

  As intermediate steps for the proof of Theorem~\ref{the:hard-edp-fvs} 
  we establish the \W{1}-hardness of
  several interesting variants of a multidimensional version of
  the well-known \textsc{Subset Sum} problem as well as several directed and
  undirected versions of EDP, which we believe are interesting in
  their own right. Namely, the proof starts by showing
  \W{1}-hardness for the following problem using a reduction from the
  well-known {\sc Multi-Colored Clique (MCC)} problem~\cite{DowneyFellows13}.
  
  \pbDefP{\textsc{Multidimensional Subset Sum (MSS)}}{An integer $k$, a set $S=\{s_1,\dotsc,s_n\}$ of item-vectors with $s_i \in \Nat^{k}$ for every $i$ with $1\leq i \leq n$ and a target vector $t \in \Nat^k$.}{$k$}{Is there a subset $S' \subseteq S$ such that $\sum_{s \in S'}s=t$?}

  Using a reduction from \textsc{MSS}, the proof then continues by establishing 
  \W{1}-hardness for the following more relaxed version of
  \textsc{MSS}.

  \pbDefP{\textsc{Multidimensional Relaxed Subset Sum (MRSS)}}{An
  integer $k$, a set $S=\{s_1,\dotsc,s_n\}$ of item-vectors with $s_i
  \in \Nat^{k}$ for every $i$ with $1\leq i \leq n$, a target vector
  $t \in \Nat^k$, and an integer $k'$.}{$k+k'$}{Is there a subset $S'
  \subseteq S$ with $|S'|\leq k'$ such that $\sum_{s \in S'}s\geq t$?}

Next, we reduce from \textsc{MRSS} to the following directed variant of EDP.
\pbDefP{\textsc{Multiple Directed Edge Disjoint Paths (MDEDP)}}{A directed
  graph $G$, a set $P$ of $\ell$ triples $(s_i,t_i,n_i)$ with $1 \leq
  i \leq \ell$, $s_i,t_i \in V(G)$, and $n_i \in \Nat$.}{$\fvs(\und{G})+|P|$}{Is there a set
  of pairwise edge disjoint paths containing $n_i$ paths from
  $s_i$ to $t_i$ for every $i$ with $1\leq i \leq \ell$?}
In the above, $\und{G}$ denotes the underlying undirected graph of a given
directed graph $G$.

Using a known result that allows to reduce the directed EDP
problem to the undirected EDP problem given by Vygen~\cite{vygen1995np}, we
then show that the undirected variant of the above problem, which we refer to as
\textsc{MUEDP}, of \textsc{MDEDP} is also \W{1}-hard. The final step
of the proof is then a straightforward reduction from
\textsc{MUEDP} to the \textsc{EDP} problem parameterized by the
feedback vertex set number of the augmented graph.
}

\lv{
In this section we will show that EDP parameterized by the feedback
vertex set number (and hence also parameterized by the treewidth)
of the augmented graph is \W{1}-hard.
This nicely complements the
result in~\cite{ZhouTN00} showing that EDP is solvable in
polynomial time for bounded treewidth and also provides a natural
justification for the fpt-algorithm presented in the next
section, since it shows that more general parameterizations such as
treewidth are unlikely to lead to an fpt-algorithm for EDP.

On the way towards this
result, we provide hardness results for several interesting versions
of the multidimensional subset sum problem (parameterized by the
number of dimensions) which we believe are interesting in their own right.
In particular, we note that the hardness results also hold for the well-known and more general
multidimensional knapsack problem~\cite{Freville04}.
Our first auxiliary result shows hardness for the following problem.

\pbDefP{\textsc{Multidimensional Subset Sum (MSS)}}{An integer $k$, a set $S=\{s_1,\dotsc,s_n\}$ of item-vectors with $s_i \in \Nat^{k}$ for every $i$ with $1\leq i \leq n$ and a target vector $t \in \Nat^k$.}{$k$}{Is there a subset $S' \subseteq S$ such that $\sum_{s \in S'}s=t$?}

\begin{LEM}
  \textsc{MSS} is \W{1}-hard even if all integers in the input are
  given in unary.
\end{LEM}
\begin{proof}
  We prove the lemma by a  parameterized reduction from
  \textsc{Multicolored Clique}, which is well-known to be 
  \W{1}\hy complete~\cite{Pietrzak03}.
  Given an integer $k$ and a
  $k$-partite graph $G$ with partition $V_1,\dotsc,V_k$, 
  the \textsc{Multicolored Clique} problem
  ask whether $G$ contains a $k$-clique. In the following we denote by
  $E_{i,j}$ the set of all edges in $G$ with one endpoint in $V_i$ and
  the other endpoint in $V_j$, for every $i$ and $j$ with $1\leq i < j
  \leq k$.
  To show the lemma, we will construct an instance $\III=(k',S,t)$ of
  \textsc{MSS} in polynomial time with $k'=2\binom{k}{2}+k$ and all integers in
  $\III$ are bounded by a polynomial in $|V(G)|$ such that $G$ has a
  $k$-clique if and only if $\III$ has a solution.

  For our reduction we will employ so called Sidon sequences of
  natural numbers. A \emph{Sidon
    sequence} is a sequence of natural numbers such that the sum of every two
  distinct numbers in the sequence is unique. For our reduction we
  will need a Sidon sequence of $|V(G)|$ natural numbers, i.e.,
  containing one number for each vertex of $G$. Since the numbers in
  the Sidon sequence will be used as numbers in $\III$, we need
  to ensure that the largest of these numbers is bounded by a
  polynomial in $|V(G)|$. Indeed~\cite{ErdosTuran41} shows that a Sidon sequence
  containing $n$ elements and whose largest element is at most $2p^2$,
  where $p$ is the smallest prime number larger or equal to $n$, can be
  constructed in polynomial time. Together with Bertrand's
  postulate~\cite{AignerZiegler10}, which states that for every natural
  number $n$ there is a prime number between $n$ and $2n$, we obtain
  that a Sidon sequence containing $|V(G)|$ numbers and whose largest
  element is at most $8|V(G)|^2$ can be found in polynomial time.
  In
  the following we will
  assume that we are given such a Sidon sequence $\SSS$ and we denote
  by~$\SSS(i)$ the~$i$\hy th element of $\SSS$ for any $i$ with $1 \leq i
  \leq |V(G)|$. Moreover, we denote by $\max(\SSS)$ and $\max_2(\SSS)$
  the largest element of~$\SSS$ and the maximum sum of any
  two numbers in~$\SSS$, respectively. We will furthermore assume that the vertices
  of $G$ are identified by numbers between $1$ and $|V(G)|$ and
  therefore $\SSS(v)$ is properly defined for every $v \in V(G)$.

  We are now ready to construct the instance $\III=(k',S,t)$.
  We set $k'=2\binom{k}{2}+k$ and $t$ is the vector whose first
  $\binom{k}{2}$ entries are all equal to $\max_2(\SSS)+1$ and whose remaining
  $\binom{k}{2}+k$ entries are all equal to $1$. For every $i$ and $j$
  with $1 \leq i < j \leq k$, we will use $I(i,j)$ as a means of enumerating the indices in a
  sequence of two-element tuples; formally, $I(i,j)=(\sum_{l=1}^{l<i}(k-l))+(j-1)$. 
  Note that the vector $t$ and its indices can then be visualized as follows:
  \[
    t=(\underbrace{{\max_2}(\SSS)+1, \dotsc, \max_2(\SSS)+1}_{I(1,2),
      I(1,3), \dotsc, I(k-1,k)}, \underbrace{1, \dotsc, 1}_{\binom{k}{2}+I(1,2),
      \binom{k}{2}+I(1,3), \dotsc, \binom{k}{2}+I(k-1,k)}, \underbrace{1, \dotsc, 1}_{2\binom{k}{2}+1,\dotsc, 2\binom{k}{2}+k)})
  \]
  
  We now proceed to the construction of $S$, which will contain one element for each
  edge and for each vertex in $G$. In particular, the set $S$ of
  item-vectors contains the following elements:
  \begin{itemize}
  \item for every $i$ with $1 \leq i \leq k$ and every $v \in V_i$,
    a vector $s_v$ such that all
    entries with index in $\SB I(l,r) \SM 1 \leq l < r \leq k \land l=i \SE\cup \SB
    I(l,r) \SM 1\leq l < r \leq k \land r=i\SE$ are equal to $\SSS(v)$ (informally, this corresponds to all indices where at least one element of the tuple $(l,r)$ is equal to $i$), the $2\binom{k}{2}+i$-th
    entry is equal to $1$, and all other entries are equal to
    $0$,
  \item for every $i$ and $j$ with $1\leq i < j \leq k$ and every $e=\{u,v\}
    \in E(i,j)$, a vector $s_e$ such that the entry $I(i,j)$ is equal
    to $(\max_2(\SSS)+1)-(\SSS(u)+\SSS(v))$, the
    $\binom{k}{2}+I(i,j)$-th entry
     is equal to $1$, and all other entries are equal
    to $0$.
  \end{itemize}
  This completes the construction of $\III$. It is clear that $\III$
  can be constructed in polynomial time and moreover every integer in
  $\III$ is at most $\max_2(\SSS)+1$ and hence polynomially bounded in
  $|V(G)|$. Intuitively, the construction relies on the fact that since the sum of each pair of vertices is unique, we can uniquely associate each pair with an edge between these vertices whose value will then be the remainder to the global upper-bound of $\max_2(\SSS)$.

  It remains to show that $G$ has $k$-clique if and only if $\III$ has
  a solution. Towards showing the forward direction, let $C$ be a
  $k$-clique in $G$ with vertices $v_1,\dotsc,v_k$ such that $v_i \in
  V_i$ for every $i$ with $1\leq i \leq k$. We claim that the
  subset $S'=\SB s_v \SM v \in V(C) \SE \cup \SB s_e \SM e \in E(C) \SE$
  of $S$ is a solution for $\III$. Let $t'$ be the vector $\sum_{s \in
    S'}s$. Because $C$ contains exactly one
  vertex from every $V_i$ and exactly one edge from every $E_{i,j}$,
  it holds that $t'[l]=t[l]=1$ for every index $l$ with
  $\binom{k}{2}< l \leq 2\binom{k}{2}+k$. Moreover, for every $i$ and $j$ with
  $1 \leq i < j \leq k$, the vectors $s_{v_i}$, $s_{v_j}$, and
  $s_{e_{i,j}}$ are the only vectors in $S'$ with a non-zero entry at
  the $I(i,j)$-th position. Hence
  $t'[I(i,j)]=s_{v_i}[I(i,j)]+s_{v_j}[I(i,j)]+s_{e_{i,j}}[I(i,j)]$,
  which because $s_{v_i}[I(i,j)]=\SSS(v_i)$,
  $s_{v_j}[I(i,j)]=\SSS(v_j)$, and
  $s_{e_{i,j}}[I(i,j)]=(\max_2(\SSS)+1)-(\SSS(v_i)+\SSS(v_j))$ is
  equal to
  $\SSS(v_i)+\SSS(v_j)+(\max_2(\SSS)+1)-(\SSS(v_i)+\SSS(v_j))=\max_2(\SSS)+1=t[I(i,j)]$,
  as required.
  
  Towards showing the reverse direction, let $S'$ be a subset of $S$
  such that $\sum_{s \in S'}s=t$. Because the last $k$ entries of $t$
  are equal to $1$ and for every $i$ with $1 \leq i \leq k$, it holds
  that the only vectors in $S$ that have a non-zero entry at the
  $i$-th last position are the vectors in $\SB s_v \SM v \in V_i\SE$,
  it follows that $S'$ contains exactly one vector say $s_{v_i}$
  in $\SB s_v \SM v \in V_i\SE$ for every $i$ with $1 \leq i \leq k$.
  Using a similar argument for the
  entries of $t$ with indices between $\binom{k}{2}+1$ and
  $2\binom{k}{2}$, we obtain that $S'$ contains exactly one vector say
  $e_{i,j}$ in $\SB s_e \SM e \in E_{i,j}\SE$ for every $i$ and $j$
  with $1 \leq i < j \leq k$. Consequently,
  $S'=\{s_{v_1},\dotsc,s_{v_k}\}\cup \SB e_{i,j} \SM 1 \leq i < j \leq
  k \SE$. We claim that $\{v_1,\dotsc,v_k\}$ forms a $k$-clique in
  $G$, i.e., for every $i$ and $j$ with $1 \leq i < j \leq k$, it
  holds that $e_{i,j}=\{v_i,v_j\}$. To see this consider the $I(i,j)$-th
  entry of $t'=\sum_{s\in S'}s$. The only vectors in $S'$ having a
  non-zero contribution towards $t'[I(i,j)]$ are the vectors
  $s_{v_i}$, $s_{v_j}$, and $s_{e_{i,j}}$. Because
  $s_{v_i}[I(i,j)]=\SSS(v_i)$, $s_{v_j}[I(i,j)]=\SSS(v_j)=\SSS(v_j)$,
  and $t'[I(i,j)]=t[I(i,j)]=\max_2(\SSS)+1$, we obtain that
  $s_{e_{i,j}}[I(i,j)]=(\max_2(\SSS)+1)-(\SSS(v_i)+\SSS(v_j))$. Because
  $\SSS$ is Sidon sequence and thus the sum $(\SSS(v_i)+\SSS(v_j))$ is
  unique, we obtain that $e_{i,j}=\{v_i,v_j\}$, as required.
\end{proof}
Observe that because any solution $S'$ of the constructed instance in
the previous lemma must be of size exactly $k'=2\binom{k}{2}+k$, it
follows that the above proof also shows
\W{1}-hardness of the following problem.

\pbDefP{\textsc{Restricted Multidimensional Subset Sum (RMSS)}}{An integer $k$, a
  set $S=\{s_1,\dotsc,s_n\}$ of item-vectors with $s_i \in \Nat^{k}$
  for every $i$ with $1\leq i \leq n$, a target vector $t \in
  \Nat^k$, and an integer $k'$.}{$k+k'$}{Is there a subset $S'
  \subseteq S$ with $|S'|=k'$ such that $\sum_{s \in S'}s=t$?}

Using an fpt-reduction from the above problem, we will now show that
also the following more relaxed version is \W{1}-hard.
\pbDefP{\textsc{Multidimensional Relaxed Subset Sum (MRSS)}}{An
  integer $k$, a set $S=\{s_1,\dotsc,s_n\}$ of item-vectors with $s_i
  \in \Nat^{k}$ for every $i$ with $1\leq i \leq n$, a target vector
  $t \in \Nat^k$, and an integer $k'$.}{$k+k'$}{Is there a subset $S'
  \subseteq S$ with $|S'|\leq k'$ such that $\sum_{s \in S'}s\geq t$?}
\begin{LEM}
  \textsc{MRSS} is \W{1}-hard even if all integers in the input are
  given in unary.
\end{LEM}
\begin{proof}
  We prove the lemma by a  parameterized reduction from
  \textsc{RMSS}. Namely, given an instance $\III=(k,S,t,k')$ of \textsc{RMSS}
  we construct an equivalent instance $\overline{\III}=(2k,\overline{S},\overline{t},k')$ of
  \textsc{MRSS} in polynomial time such that all integers in $\overline{\III}$
  are bounded by a polynomial of the integers in $\III$.

  The set $\overline{S}$ contains one vector $\overline{s}$ for every vector $s \in S$
  with $\overline{s}[i]=s[i]$ and $\overline{s}[k+i]=t[i]-s[i]$ for every $i$ with $1\leq
  i \leq k$. Finally, the target vector $\overline{t}$ is defined by setting $\overline{t}[i]=t[i]$
  and $\overline{t}[k+i]=(k'-1)\cdot t[i]$ for every $i$ with $1 \leq i \leq k$.
  This concludes the construction of $\overline{\III}$. Clearly, $\overline{\III}$ can be
  constructed in polynomial time and the values of all numbers in $\overline{\III}$ are
  bounded by a polynomial of the maximum number in $\III$.
  It remains to show that $\III$ has a solution if and only if so does
  $\overline{\III}$.

  Towards showing the forward direction, let $S' \subseteq S$ be a
  solution for $\III$, i.e., $|S'|=k'$ and $\sum_{s\in S'}s=t$. We
  claim that the set $\overline{S}'=\SB \overline{s} \SM s\in S'\SE$
  is a solution for $\overline{\III}$. Because $\overline{s}[i]=s[i]$
  and $\overline{t}[i]=t[i]$
  for every $s \in S$ and $i$ with $1 \leq i \leq k$, it follows that
  $\sum_{\overline{s}\in
    \overline{S}'}\overline{s}[i]=\overline{t}[i]$ for every $i$ as
  above. Moreover, for every $i$ with $1 \leq i \leq k$, it holds that
  $\sum_{\overline{s}\in
    \overline{S}'}\overline{s}[k+i]=k'\cdot t[i] - \sum_{\overline{s}\in
    \overline{S}'}\overline{s}[i]=k'\cdot t[i] -
  t[i]=(k'-1)t[i]=\overline{t}[k+i]$, showing that $\overline{S}$ is a
  solution for $\overline{\III}$.

  Towards showing the reverse direction, let $\overline{S}' \subseteq
  \overline{S}$ be a solution for $\III'$, i.e., $|\overline{S}'|\leq k'$ and
  $\sum_{\overline{s}\in \overline{S}'}\overline{s}\geq
  \overline{t}$. We claim that the set $S'=\SB s \SM \overline{s}\in
  \overline{S}'\SE$ is a solution for $\III$.

  Because $\overline{S}'$ is a solution for $\III'$, we obtain for
  every $i$ with $1 \leq i \leq k$ that:
  \begin{itemize}
  \item[(1)] $\sum_{s \in \overline{S}'}\overline{s}[i]\geq t[i]$, which
    because $s[i]=\overline{s}[i]$ implies that $\sum_{s \in
      S'}s[i]\geq t[i]$,
  \item[(2)] $\sum_{\overline{s} \in \overline{S}'}\overline{s}[k+i]\geq
    (k-1)t[i]$, which because $\overline{s}[k+i]=t[i]-s[i]$ implies that $|S'|t[i]-\sum_{s\in S'}s[i]\geq (k'-1)t[i]$. First, since $\sum_{s\in S'}s[i]>0$ by (1), observe that $|S'|> k'-1$ and in particular $|S'|=k'$. Then by using this, we obtain that $k't[i]-\sum_{s\in S'}s[i]\geq (k'-1)t[i]$ which implies $t[i]\geq \sum_{s\in S'}s[i]$.
  \end{itemize}
  It follows from (1) and (2) that $\sum_{s\in S'}s[i]=t[i]$ and hence
  $S'$ is a solution for $\III$ of size $k'$, as required.
\end{proof}

Our next step is to move from subset sum problems to EDP variants, with a final goal
of showing the hardness of EDP parameterized by the feedback vertex set number
of the augmented graph. We will first show that the following directed
version, which also allows for multiple instead of just one path between
every terminal pair, is \W{1}-hard parameterized by the feedback
vertex set number of the augmented graph and the number of terminal pairs.

\pbDef{\textsc{Multiple Directed Edge Disjoint Paths (MDEDP)}}{A directed
  graph $G$, a set $P$ of $\ell$ triples $(s_i,t_i,n_i)$ with $1 \leq
  i \leq \ell$, $s_i,t_i \in V(G)$, and $n_i \in \Nat$.}{Is there a set
  of pairwise edge disjoint paths containing $n_i$ paths from
  $s_i$ to $t_i$ for every $i$ with $1\leq i \leq \ell$?}

\begin{LEM}\label{lem:hard-mdedp}
  MDEDP is \W{1}-hard parameterized by the following parameter: the
  feedback vertex set number of the undirected augmented graph plus
  the number of terminal triples. Furthermore, this holds even for acyclic instances when all integers on the input are given in unary.
\end{LEM}
\begin{proof}
  We prove the lemma by a parameterized reduction from
  \textsc{MRSS}. Namely, given an instance $\III=(k,S,t,k')$ of
  \textsc{MRSS} we construct an equivalent instance
  $\III'=(G,P)$
  of \textsc{MDEDP} in polynomial time such
  that $|P|\leq k+1$, $\fvs(G^P)\leq 2k+2$ and $G$ is acyclic.
  
  We start by introducing a gadget $G(s)$ for every item-vector $s
  \in S$. $G(s)$ is a directed path $(p_1^s,\dotsc,p_l^s)$, where $l=2(\sum_{1 \leq i \leq
    k}s[i])+2$. We let $P$ contain one triple
  $(s,t,|S|-k')$ as well as one triple $(s_i,t_i,t[i])$ for every
  $i$ with $1\leq i \leq k$. Then $G$ is obtained from the disjoint
  union of $G(s)$ for every $s \in S$ plus the vertices
  $\{s,t,s_1,t_1,\dotsc,s_k,t_k\}$. Moreover, $G$ contains the
  following edges:
  \begin{itemize}
  \item one edge from $s$ to the first vertex of $G(s)$ for every $s
    \in S$,
  \item one edge from the last vertex of $G(s)$ to $t$ for every $s
    \in S$,
  \item for every $i$ with $1\leq i \leq k$ and every $s \in S$ an
    edge from $s_i$ to the vertex $p_j^s$ (in $G(s)$), where
    $j=1+2(\sum_{1\leq j < i}s[j])+2l-1$ for every $l$ with $1 \leq l
    \leq s[i]$,
  \item for every $i$ with $1\leq i \leq k$ and every $s \in S$ an
    edge from $p_j$ (in $G(s)$), where
    $j=1+2(\sum_{1\leq j < i}s[j])+2l$, to $t_i$ for every $l$ with $1
    \leq l \leq s[i]$,
  \end{itemize}
  This completes the construction of $(G,P)$. Since $G^P -
  \{s,t,s_1,t_1,\dotsc,s_k,t_k\}$ is a disjoint union of directed
  paths, i.e., one path $G(s)$ for every $s \in S$, we obtain that
  $G^P$ has a feedback vertex set, i.e., the set
  $\{s,t,s_1,t_1,\dotsc,s_k,t_k\}$,
  of size at most $2(k+1)=2k+2$. Moreover, $G$
  is clearly acyclic and $|P|\leq k+1$. It hence only remains to
  show that $(k,S,t,k')$ has a solution if and only if so does
  $(G,P)$.
  
  Towards showing the forward direction let $S' \subseteq S$ be a
  solution for $\III$. Then we construct a set $Q$ of pairwise
  edge-disjoint paths in $G$ containing $|S|-k'$ path from $s$ to
  $t$ as well as $t[i]$ paths from $s_i$ to $t_i$, i.e., a solution
  for $\III'$, as follows. For every $s \in S
  \setminus S'$, $Q$ contains the path $(s,G(s),t)$, which already
  accounts for the $|S|-k'$ paths from $s$ to $t$. Moreover,
  for every $i$ with $1 \leq i \leq k$ and $s \in S'$, $Q$ contains
  the path $(s_i,p_{j},p_{j+1},t_i)$ for every $j$ with
  $1+2(\sum_{1\leq j < i}s[j])< j \leq 1+2(\sum_{1\leq j <
    i}s[j])+2(s[i])$ and $j$ is even.
  This concludes the definition of $Q$. Note that $Q$ now contains
  $\sum_{s \in S'}s[i]$ paths from $s_i$ to $t_i$ for every $i$ with
  $1 \leq i \leq k$ and since $S'$ is a solution for $\III$, it
  holds that $\sum_{s \in S'}s[i]\geq t[i]$. Consequently, $Q$ is a
  solution for $\III'$.
  
  Towards showing the reverse direction, let $Q$ be a solution for
  $\III'$. Then $Q$ must contain $|S|-k'$ pairwise edge
  disjoint paths from $s$ to $t$. Because every path from $s$ to
  $t$ in $G$ must have the form $(s,G(s),t)$ for some $s \in S$,
  we obtain that all the edges of exactly $|S|-k'$ gadgets $G(s)$ for $s \in
  S$ are used by the paths from $s$ to $t$ in $G$. Let $S' \subseteq
  S$ be the set containing all $s \in S$ such that $G(s)$ is not
  used by a path from $s$ to $t$ in $Q$. We claim that $S'$ is a
  solution for $\III$. Clearly, $|S'|=k'$ and it remains to show that
  $\sum_{s \in S'}s\geq t$. Towards showing this observe that for
  every $i$ with $1 \leq i \leq k$, a path from $s_i$ to $t_i$ has
  to use at least one edge $\{p_j^s,p_{j+1}^s\}$ from some $s \in
  S'$ and $j$ with $1+2(\sum_{1 \leq j < i}s[i])< j \leq 1+2(\sum_{1
    \leq j \leq i}s[i])$ and $j$ is odd. Since $G(s)$ for every $s
  \in S'$ has at most $s[i]$ of those edges, we obtain that for
  every $s \in S'$, $Q$ contains at most $s[i]$ path from $s_i$ to
  $t_i$ using edges in $G(s)$. Because the paths in $Q$ from $s_i$
  to $t_i$ cannot use any edge from $G(s)$ such that $s \notin S'$;
  this is because all edges of such $G(s)$ are already used by the
  paths from $s$ to $t$ in $Q$. Hence the total number of paths in
  $Q$ from $s_i$ to $t_i$ is at most $\sum_{s\in S'}s[i]$ and since
  $Q$ contains $t[i]$ paths from $s_i$ to $t_i$, we obtained that
  $\sum_{s\in S'}s[i]\geq t[i]$.
\end{proof}

Our next aim is now to reduce from \textsc{MDEDP} to the following
undirected version of \textsc{MDEDP}.

\pbDef{\textsc{Multiple Undirected Edge Disjoint Paths (MUEDP)}}{An undirected
  graph $G$, a set $P$ of $\ell$ triples $(s_i,t_i,n_i)$ with $1 \leq
  i \leq \ell$, $s_i,t_i \in V(G)$, and $n_i \in \Nat$.}{Is there a set
  of pairwise edge disjoint paths containing $n_i$ paths from
  $s_i$ to $t_i$ for every $i$ with $1\leq i \leq \ell$?}

To do so we first need the following auxiliary lemma.
\begin{LEM}\label{lem:mdedp-euler}
  Let $\III=(G,P)$ be an instance of \textsc{MDEDP} such that $G$ is
  acyclic. Then in polynomial time we can construct an instance
  $\III'=(G',P')$ such that:
  \begin{itemize}
  \item[(P1)] $\III$ has a solution if and only if so does $\III'$,
  \item[(P2)] $G'$ is acyclic and the graph $G'(P')$ is Eulerian, where $G'$
    is the graph obtained from $G'$ after adding $n$ edges from
    $t$ to $s$ for every $(s,t,n) \in P$,
  \item[(P3)] $|P'|\leq |P|+1$ and $\fvs(\und{G'}) \leq \fvs(\und{G})+2$.
  \end{itemize}
\end{LEM}
\begin{proof}
  The construction of $\III'$ is based on the construction of
  $(G',H')$ from $(G,H)$ in~\cite[Theorem 2]{vygen1995np}.
  Namely, for every $v \in V(G)$ let
  $\alpha(v)=\max\{0,\delta_{G(P)}(v)-\rho_{G(P)}(v)\}\}$
  and $\beta(v)=\max\{0,\rho_{G(P)}(v)-\delta_{G(P)}(v)\}\}$, where
  $\delta_D(v)$ and $\rho_D(v)$ denote the number of out-neighbors
  respectively in-neighbors of a vertex $v$ in a directed graph $D$.
  . Then
  \begin{eqnarray*}
    0 & = & |E(G(P))|-|E(G(P))|\\
      & = & \sum_{v\in V(G)}(\delta_{G(P)}(v)-\rho_{G(P)}(v))\\
      & = & \sum_{v\in V(G)}(\alpha(v)-\beta(v))
  \end{eqnarray*}
  Hence $\sum_{v\in V(G)}\alpha(v)=\sum_{v\in V(G)}\beta(v)=q$. We now
  construct the instance $\III'=(G',P')$ from $(G,P)$ by adding one
  triple $(s,t,q)$ to $P$ as well as adding the vertices $s$ and $t$
  together with $\alpha(v)$ edges from $s$ to $v$ and $\beta(v)$ edges
  from $v$ to $t$ for every $v \in V(G)$.
  It is straightforward to verify that $\III'$ satisfies Properties
  (P2) and (P3). Moreover, the reverse direction of Property (P1) is
  trivial. Towards showing the forward direction of Property (P1),
  assume that $\III$ has a solution $S$ and let $G''$ be the graph
  obtained from $G'$ after removing all edges appearing in a path in
  $S$. Then $G''(\{(s,t,q)\})$ is Eulerian and $G''$ is acyclic, hence
  there are $q$ pairwise edge-disjoint paths in $G''$ from $s$ to $t$.
\end{proof}

We are now ready to show that \textsc{MUEDP} is \W{1}-hard
parameterized by the combined parameter the feedback vertex set number
of the augmented graph and the number of terminal triples.
\begin{LEM}\label{lem:hard-muedp}
  MUEDP is \W{1}-hard parameterized by the following parameter: the
  feedback vertex set number of the augmented graph plus
  the number of terminal triples. Furthermore, this holds even when all integers on the input are given in unary.
\end{LEM}
\begin{proof}
  We prove the lemma by a parameterized reduction from
  \textsc{MDEDP}. Namely, given an instance $\III=(G,P)$ of
  \textsc{MDEDP} we construct an equivalent instance
  $\III'=(H,Q)$
  of \textsc{MUEDP} in polynomial time such
  that $|Q|\leq |P|+1$, $\fvs(H^Q)\leq \fvs(G^P)$. The result then
  follows from Lemma~\ref{lem:hard-mdedp}.

  Let $(G',P')$ be the instance of \textsc{MDEDP} obtained from
  $(G,P)$ by Lemma~\ref{lem:mdedp-euler}. Then $\III'=(H,P)$ is simply
  obtained from $(G',P')$ by disregarding the directions of all edges
  in $G'$. Because of Lemma~\ref{lem:mdedp-euler}, we obtain that
  $|Q|\leq |P|+1$ and $\fvs(\und{H^Q})\leq \fvs(\und{G^P})$ and it hence only
  remains to show the equivalence of $\III$ and $\III'$.
  Note that because of Lemma~\ref{lem:mdedp-euler}, it holds that
  $(G',P')$ is acyclic and $G'(P')$ is Eulerian. It now
  follows from~\cite[Lemma 5]{vygen1995np} that $(G',P')$
  and $(H,P)$ are equivalent, i.e., any solution for $(G',P')$ is a
  solution for $(H,P)$ and vice versa. Together with the equivalency
  between $(G,P)$ and $(G',P')$ (due to Lemma~\ref{lem:mdedp-euler})
  this shows that $\III$ and $\III'$ are
  equivalent and concludes the proof of the lemma.
\end{proof}

Finally, using a very simple reduction from \textsc{MUEDP} we are
ready to show that \textsc{EDP} is \W{1}-hard parameterized by the
feedback vertex set number of the augmented graph.
\begin{THE}
  EDP is \W{1}-hard parameterized by the feedback vertex set number
  of the augmented graph.
\end{THE}
\begin{proof}
  We prove the lemma by a parameterized reduction from
  \textsc{MUEDP}. Namely, given an instance $\III=(G,P)$ of
  \textsc{MUEDP} we construct an equivalent instance
  $\III'=(G',P')$
  of \textsc{EDP} in polynomial time such
  that $\fvs(\und{{G'}^{P'}})\leq \fvs(\und{G^P})+2|P|$. The result then
  follows from Lemma~\ref{lem:hard-muedp}.

  $\III'$ is obtained from $\III$ as follows. For every $(s,t,n) \in
  P$, we add $2n$ new vertices $s^1,\dotsc,s^n$ and $t^1,\dotsc,t^n$
  to $G$ and an edge between $s^i$ and $s$ as well as an edge between
  $t^i$ and $t$ for every $1 \leq i \leq n$. Moreover, for every $i$
  with $1 \leq i \leq n$, we add the terminal pair $(s^i,t^i)$ to
  $P'$. It is straightforward to verify that $\III$ and $\III'$ are
  equivalent. Moreover, if $F$ is a feedback vertex set for $G^P$,
  then $F \cup \SB s,t \SM (s,t,n) \in G \SE$ is a feedback vertex set
  for ${G'}^{P'}$ and hence $\fvs(\und{{G'}^{P'}})\leq \fvs(\und{G^P})+2|P|$, which
  concludes the proof of the theorem.
\end{proof}
}

\section{An fpt-Algorithm for EDP using the Augmented Graph}
In light of Theorem~\ref{the:hard-edp-fvs}, it is natural to ask whether there exist natural structural parameters of the augmented graph which would give rise to fpt-algorithms for EDP but which cannot be used on the input graph. In other words, does considering the augmented graph instead of the input graph provide any sort of advantage in terms of fpt-algorithms? In this section we answer this question affirmatively by showing that EDP is fixed-parameter tractable
parameterized by the so-called \emph{fracture number} of the augmented graph. We note that a parameter similar to the fracture number has recently
been used to obtain fpt-algorithms for Integer Linear
Programming~\cite{DEGetalIJCAI17}.


\begin{definition}
  A vertex subset $X$ of a graph $H$ is called a \emph{fracture
    modulator} if each connected component in $H\setminus V$ contains at most
  $|X|$ vertices. We denote the size of a minimum-cardinality fracture
  modulator in $H$ as $\text{frac}(H)$ or the \emph{fracture number}
  of $H$.
\end{definition}
\lv{Note the the fracture number is always at most the treewidth of the
augmented graph and even though it is known that EDP parameterized by
the treewidth of the augmented graph is in XP~\cite{ZhouTN00}, it open whether it is
in FPT. Moreover, \cite[Theorem 6]{FleszarMS16} shows that EDP
parameterized by the fracture number of the input graph is paraNP-hard.
}
We begin by making a simple structural observation about fracture modulators.

\sv{\begin{lemma}[$\star$]}
\lv{\begin{lemma}}
  Let $(G,P)$ be an instance of \textsc{EDP} and let $k$ be the fracture
  number of its augmented graph. Then there exists a fracture modulator
  $X$ of $G^P$ of size at most $2k$ such that $X$ does not contain any
  terminal vertices. Furthermore, such a fracture modulator $X$ can be
  constructed from any fracture modulator of size at most $k$ in linear
  time.
\end{lemma}
\lv{
\begin{proof}
  Let $X'$ be arbitrary fracture modulator of size $k$ and let $P'
  \subseteq P$ be the set of terminal pairs $p$ such that $p \cap
  X'\neq \emptyset$. Consider the set $X=X' \setminus \bigunion{P'}
  \cup \SB N_G(a) \SM a \in \bigunion{P'}\SE$. Because every terminal
  $a \in \bigunion{P}$ has at most one neighbor in $G$ (recall our
  assumption on $(G,P)$ given in Subsection~\ref{ssec:edp}), it holds
  that $|X|\leq 2|X'|$. Moreover, for vertex $a$ that we removed from
  $X'$, it holds that $a$ is in a component of size two in $G^P
  \setminus X$, i.e., the component consisting of $a$ and $b$ where
  $\{a,b\} \in P'$. Consequently, every component of $G^P - X$
  either has size two or it is a subset of a component of $G^P- X'$.
\end{proof}
}
\sv{We note that the problem of computing a fracture modulator of size at
most $k$ has been recently considered in the context of Integer Linear Programming~\cite{DEGetalIJCAI17}.}

\lv{We note that the problem of computing a fracture modulator of size at
most $k$ is closely related to the  \textsc{Vertex Integrity} problem~\cite{DrangeDregiHof16}, and that a variant of it has been recently considered in the context of Integer Linear Programming~\cite{DEGetalIJCAI17}. }

\begin{lemma}[{\cite[Theorems 7 and 8]{DEGetalIJCAI17}}]
\label{lem:computefracture}
There exists an algorithm which takes as input a graph $G$ and an
integer $k$, runs in time at most   $\bigoh((k+1)^k|E(G)|)$, and
outputs a fracture modulator of cardinality at most $k$ if such
exists. Moreover, there is a polynomial-time algorithm that either
computes a fracture modulator of size at most $(k+1)k$ or outputs
correctly that no fracture modulator of size at most $k$ exists.
\end{lemma}
\lv{
\begin{proof}
The algorithm is based on a bounded search tree approach and relies on the following
  observations.
  \begin{description}
  \item[(O1)] If $G$ is not connected then each fracture modulator of $G$ is the disjoint union of fracture modulators for all connected components of $G$.
  \item[(O2)] If $G$ is connected and $D$ is any set of $k+1$ vertices
    of $G$ such that $G[D]$ is connected, then any fracture modulator 
    has to contain at least one vertex from $D$. 
  \end{description}
  These observations lead directly to the following recursive algorithm
  that either determines that the instance $(G,k)$ is a NO-instance or outputs a
  fracture modulator $X$ of size at most $k$. The algorithm also
  remembers the maximum size of any component in a global constant
  $c$, which is set to $k$ for the whole duration of the algorithm.
  The algorithm first checks
  whether $G$ is connected. If $G$ is not connected the algorithm
  calls itself recursively on the instance $(G[C], k)$ for each connected component~$C$
  of $G$. If one of the recursive calls returns NO or if the size of
  the union of the solutions returned for each component exceeds $k$,
  the algorithm returns that $I$ is a NO-instance. Otherwise the
  algorithm returns the union of the solutions returned for each
  component of $G$. 

  If $G$ is connected and $|V(G)|\leq c$, the algorithm returns the
  empty set as a solution. Otherwise, i.e., if $G$ is connected but
  $|V(G)|>c$ the algorithm first computes a set $D$ of $c+1$ vertices 
  of $G$ such that $G[D]$ is connected. This can for instance be achieved
  by a depth-first search that starts at any vertex of $G$ and stops
  as soon as $c+1$ vertices have been visited. The algorithm then
  branches on the vertices in $D$, i.e., for every $v \in D$ the
  algorithm recursively computes a solution for the instance $(G
  -\{v\},k-1)$. It then returns the solution of minimum size
  returned by any of those recursive calls, or NO if none of those
  calls return a solution. This completes the description of the
  algorithm. The correctness of the algorithm follows immediately from
  the above observations. Moreover the running time of the algorithm
  is easily seen to be dominated by the maximum time required for the
  case that at each step of the algorithm $G$ is connected.
  In this case the running time can be obtained as the
  product of the number of branching steps times the time spent on
  each of those. Because at each recursive call the parameter $k$ is
  decreased by at least one and the number of branching choices is at
  most $c+1$, we obtain that there are at most $(c+1)^k=(k+1)^k$ branching
  steps. Furthermore, the time at each branching step is dominated by
  the time required to check whether $G$ is connected, which is linear
  in the number of edges of $G$. Putting everything together,
  we obtain $\bigoh((k+1)^k|E(G)|)$ as the total time required by the
  algorithm, which completes the proof of the lemma.
\end{proof}
We note that the depth-first search algorithm in the above proof can be
easily transformed into a polynomial time approximation algorithm for
finding fracture modulators, with an approximation ratio of
$k+1$. In particular, instead of branching on the vertices of
a connected subgraph $D$ of $G$ with $k+1$ vertices, this algorithm
would simply add all the vertices of $D$ into the current
solution.
}

For the rest of this section, let us fix an instance $(G,P)$ of \textsc{EDP} with a fracture modulator $X$ of $G^P$ of cardinality $k$ which does not contain any terminals. Furthermore, since the subdivision of any edge (i.e., replacing an edge by a path of length $2$) in $(G,P)$ does not change the validity of the instance, we will assume without loss of generality that $G[X]$ is edgeless; in particular, any edges that may have had both endpoints in $X$ will be subdivided, creating a new connected component of size $1$.

Our next step is the definition of \emph{configurations}. These capture one specific way a connected component $C$ of $G^P-X$ may interact with the rest of the instance. 
It will be useful to observe that for each terminal pair $ab$ there exists precisely one connected component $C$ of $G^P-X$ which contains both of its terminals; we say that $ab$ \emph{occurs} in $C$. For a connected component $C$, we let $C^+$ denote the induced subgraph on $G^P[C\cup X]$.

A \emph{trace} is a tuple $(x_1,\dots,x_\ell)$ of elements of $X$. 
A configuration is a tuple $(\alpha,\beta)$ where 
\begin{itemize}
\item $\alpha$ is a multiset of at most $k$ traces, and 
\item $\beta$ is a mapping from subsets of $X$ of cardinality $2$ to $[k^2]$.
\end{itemize}

A component $C$ of $G^P$ \emph{admits} a configuration $(\alpha,\beta)$ if there exists a set of edge disjoint paths $\cF$ in $C^+$ 
and a surjective mapping $\tau$ (called the \emph{assignment}) from $\alpha$ to the terminal pairs that occur in $C$
with the following properties. 
\begin{itemize}
\item For each terminal pair $st$ that occurs in $C$: 
\begin{itemize}
\item $\cF$ either contains a path whose endpoints are $s$ and $t$, or 
\item $\cF$ contains an $s$-$x_1$ path  for some $x_1\in X$ and a $t$-$x_2$  path for some distinct $x_2\in X$ and there exists a trace $L=(x_1,\dots,x_2)\in \alpha$ such that $\tau(L)=st$.
\end{itemize}
\item for each distinct $a,b\in X$, $\cF$ contains precisely $\beta(\{a,b\})$ paths from $a$ to $b$.
\item $\cF$ contains no other paths than the above.
\end{itemize}

Intuitively, $\alpha$ stores information about how one particular set of edge disjoint paths $A$ which originate in $C$ is routed through the instance: they may either be routed only through $C^+$ (in which case they don't contribute to $\alpha$), or they may leave $C^+$ (in which case $\alpha$ stores the order in which these paths visit vertices of $X$, i.e., their trace). On the other hand, $\beta$ stores information about how paths that originate outside of $C$ can potentially be routed through $C$ (in a way which does not interfere with $A$). Observe that for any particular $\alpha$ there may exist several distinct configurations ($(\alpha,\beta_1),$ $(\alpha,\beta_2)$ and so forth); similarly, for any particular $\beta$ there may exist several distinct configurations ($(\alpha_1,\beta),$ $(\alpha_2,\beta)$ and so forth).

If a set $\cF$ of edge disjoint paths in $C^+$ satisfies the conditions specified above for a configuration $(\alpha,\beta)$, we say that $\cF$ \emph{gives rise} to $(\alpha,\beta)$. Clearly, given $\cF$ and $(\alpha,\beta)$, it is possible to determine whether $\cF$ gives rise to $(\alpha,\beta)$ in time polynomial in $|V(C)|$.

While configurations capture information about how a component can interact with a set of edge disjoint paths, our end goal is to have a way of capturing all important information about a component irrespective of any particular selection of edge disjoint paths. To this end, we introduce the notion of \emph{signatures}. A signature of a component $C$, denoted $\text{sign}(C)$, is the set of all configurations which $C$ admits. The set of all configurations is denoted by~$\Lambda$.

\sv{\begin{lemma}[$\star$]}
\lv{\begin{lemma}}
\label{lem:signbounds}
Given a component $C$, it is possible to compute $\text{sign}(C)$ in time at most $k^{\bigoh(k^2)}$. Furthermore, $|\text{sign}(C)|\leq |\Lambda| \leq k^{\bigoh(k^2)}$.
\end{lemma}
\lv{
\begin{proof}
We begin with the latter claim. The number of traces is $k!+(k-1)!+\dots+1!$, which is upper-bounded by $2\cdot k!$. Consequently, the number of choices for $\alpha$ is upper-bounded by $(2\cdot k!)^k\leq k^{\bigoh(k^2)}$. On the other hand, the number of choices for $\beta$ is upper-bounded by $(k^2)^{k^2}$. Since the number of configurations is upper-bounded by the number of choices for $\alpha$ times the number of choices for $\beta$, we obtain $|\Lambda|\leq k^{\bigoh(k^2)}$. Note that it is possible to exhaustively construct $\Lambda$ in the same time bound.

For the former claim, observe that $C^+$ contains at most $2k^2$ edges. Consider the following exhaustive algorithm on $C^+$. The algorithm exhaustively tries all assignments of edges in $C^+$ to labels from $\{\emptyset\}\cup [2k^2]$. For each such assignment$\cF$, it checks that each label forms a path in $C^+$; if this is the case, then $\cF$ is a collection of edge disjoint paths of $C^+$. We can then loop through each configuration $(\alpha,\beta)$ in $\Lambda$, check whether $\cF$ gives rise to $(\alpha,\beta)$, and if yes we add $(\alpha,\beta)$ to $\text{sign}(C)$. In total, this takes time at most $k^{\bigoh(k^2)}\cdot k^{\bigoh(k^2)} = k^{\bigoh(k^2)}$.
\end{proof}
}

Our next step is the formulation of a clear condition linking configurations of components in $G^P-X$ and solving $(G,P)$. This condition will be of importance later, since it will be checkable by an integer linear program. For a trace $\alpha$, we say that $a,b$ \emph{occur consecutively} in $\alpha$ if elements $a$ and $b$ occur consecutively in the sequence $\alpha$ (regardless of their order), i.e., $\alpha=(\dots,a,b,\dots)$ or $\alpha=(\dots,b,a,\dots)$. Let $\mathcal D$ be the set of connected components of $G^P-X$.

A \emph{configuration selector} is a function which maps each
connected component $C$ in $G^P-X$ to a configuration
$(\alpha,\beta)\in \text{sign}(C)$. We say that a configuration
selector $S$ is \emph{valid} if it satisfies the condition that
$\text{dem}(ab)\leq \text{sup}(ab)$ for every $\{a,b\}\subseteq X$,
where dem (\emph{demand}) and sup (\emph{supply}) are defined as
follows:
\sv{dem$(ab)$ is the number of traces in $\bigcup_{C\in \mathcal
    D}S(C)$ where $ab$ occur consecutively and sup$(ab)$ is the sum of all the values $\beta(a,b)$ in $\bigcup_{C\in \mathcal D} S(C)$. 
}
\lv{\begin{itemize}
\item dem$(ab)$ is the number of traces in $\bigcup_{C\in \mathcal D}S(C)$ where $ab$ occur consecutively. 
\item sup$(ab)$ is the sum of all the values $\beta(a,b)$ in $\bigcup_{C\in \mathcal D} S(C)$. 
\end{itemize}
}

\lv{For completeness, we provide the formal definitions of sup($ab$) and dem($ab$) below. 
\begin{itemize}
\item dem($ab$): let the multiset ${\mathcal A}^0$ be the restriction of the multiset $\bigcup_{C\in \mathcal D} S(C)$ to the first component of each configuration. Then we set ${\mathcal A}=\bigcup_{\alpha\in {\mathcal A}^0} \alpha$, i.e., ${\mathcal A}$ is the multiset of all traces which occur in configurations originating from $S$. Finally, let ${\mathcal A}_{ab}$ be the restriction of ${\mathcal A}$ only to those traces where $ab$ occurs consecutively, and we set dem$(ab)=|{\mathcal A}_{ab}|$.
\item sup($ab$): Stated formally, let the multiset ${\mathcal B}^0$ be the restriction of the multiset $\bigcup_{C\in \mathcal D} S(C)$ to the second component of each configuration. Then we set ${\mathcal B}=\bigcup_{\beta\in {\mathcal B}^0} \beta$, i.e., ${\mathcal B}$ is the multiset of all mappings which occur in configurations originating from $S$.
Finally, we set sup$(ab)=\sum_{\beta\in {\mathcal B}}\beta(ab)$.
\end{itemize}
}

The next, crucial lemma links the existence of a valid configuration selector to the existence of a solution for \textsc{EDP}.

\lv{\begin{lemma}}
\sv{\begin{lemma}[$\star$]}
\label{lem:selcorrect}
$(G,P)$ is a yes-instance if and only if there is a valid configuration selector.
\end{lemma}

\lv{
\begin{proof}
Assume $(G,P)$ is a yes-instance and let $Q$ be a solution, i.e., $Q$ is a set of edge disjoint paths in $G$ which link each terminal pair in $P$. We will construct a valid configuration selector $S$.
First, consider a connected component $C$ of $G^P-X$. Observe that $Q$
restricted to $C^+$ forms a set of edge disjoint paths $Q_C$ which
consists of:
\sv{(1) paths starting and ending in $X$,
(2) paths starting and ending at terminals of a terminal pair, and (3)
paths starting at terminals of a terminal pair and ending in $X$.
}
\lv{\begin{itemize}
\item paths starting and ending in $X$,
\item paths starting and ending at terminals of a terminal pair, and
\item paths starting at terminals of a terminal pair and ending in $X$.
\end{itemize}}

We will use $Q_C$ to construct a configuration $(\alpha_C,\beta_C)$ of $C$, as follows. For each $\{x,y\}\subseteq X$, $\beta_C(\{x,y\})$ will map each tuple $\{x,y\}$ to the number of paths in $Q_C$ whose endpoints are precisely $x$ and $y$. On the other hand, for each pair of paths in $Q_C$ which start at terminals $s,t$ of a terminal pair and end at $x,y \in X$, we add the trace $(x,z_1,\dots,z_\ell,y)$ to $\alpha$, where $z_i$ is the $(i+1)$-th vertex visited by the $s$-$t$ path of $Q$ in $X$. For example, if the $s$-$t$ path used by the solution intersects with vertices in $X$ in the order $(a,b,c,d)$, then we add the trace $(a,b,c,d)$ to $\alpha_C$. As witnessed by $Q_C$, the resulting configuration $(\alpha_C,\beta_C)$ is a configuration of $C$ and in particular $(\alpha_C,\beta_C)\in \text{sign}(C)$. 

Let $S$ be a configuration selector which maps each connected component $C$ to $(\alpha_C,\beta_C)$ constructed as above. We claim that $S$ is valid. Indeed, for each $\{a,b\}\in X$, dem$(ab)$ is the number of terminal-to-terminal paths in $Q$ which consecutively visit $a$ and then $b$ in $X$ (or vice-versa, $b$ and then $a$). At the same time, sup$(ab)$ is the number of path segments in $Q$ whose endpoints are $a$ and $b$. Clearly, dem$(ab)=\text{sup}(ab)$.

On the other hand, assume we have a valid configuration selector $S$. We will use $S$ to argue the existence of a set $Q$ of edge disjoint paths which connect all terminal pairs in $(G,P)$. For each component $C$ we know that $C$ admits $S(C)=(\alpha_C,\beta_C)$ and hence there exists a set of edge disjoint paths, say $\cF_C$, which satisfies the following conditions:
\begin{itemize}
\item For each terminal pair $st$ in $C$, $\cF_C$ either contains a path whose endpoints are $s$ and $t$ or two paths from $s$ and $t$ to two distinct endpoints in $X$;
\item For each distinct $a,b\in X$, $\cF_C$ contains precisely $\beta_C(\{a,b\})$ paths from $a$ to $b$.
\end{itemize}

Let $\cF=\bigcup_{C\in \mathcal D} \cF_C$, and let $\tau_C$ be the surjective assignment function for $C$ which goes with $\cF_C$. We will now construct the set $Q$ from $\cF$ as follows. For each terminal pair $s,t$ in some component $C$, we add a $s$-$t$ path $L$ into $Q$, where $L$ is obtained as follows. If $\cF_C$ contains a path whose endpoints are $s$ and $t$, then we set this to $L$. Otherwise, $L$ is composed of the following segments:
\begin{itemize}
\item the two paths in $\cF_C$ whose endpoints are $s$ and $t$, and
\item for each $x,y\in X$ which occur consecutively in the trace $\tau_C^{-1}(st)$, we choose an arbitrary $x$-$y$ path segment from $\cF$, use it for $L$, and then delete this path segment from $\cF$.
\end{itemize}

First of all, observe that since $S$ is valid, there will always be enough $x$-$y$ path segments in $\cF$ to choose from in the construction of $L$. Furthermore, each $L$ is an $s$-$t$ path, and all the paths in $Q$ are edge disjoint. Hence $(G,P)$ is indeed a yes-instance and the lemma holds.
\end{proof}
}

\sv{
The problem whether there is a valid configuration selector can be
easily translated into an integer linear program with a number of
variables that can be bounded in terms of $k$. It then follows
from~\cite{Lenstra83} that the problem is fixed-parameter tractable
parameterized by $k$.
}
\sv{\begin{lemma}[$\star$]}
\lv{\begin{lemma}}
\label{lem:validsel}
There exists an algorithm which takes as input an \textsc{EDP} instance $(G,P)$ and a fracture modulator $X$ of $G^P$ and determines whether there exists a valid configuration selector $S$ in time at most $2^{2^{k^{\bigoh(k^2)}}}\cdot |V(G)|$.
\end{lemma}

\lv{\begin{proof}
Consider the following instance of Integer Linear Programming. For each signature $\eta$ such that there exists a connected component $C$ of $G^P-X$ where $\text{sign}(C)=\eta$, and for each configuration $(\alpha,\beta)\in \eta$, we create a variable $z^\eta_{(\alpha,\beta)}$, and for each we add a constraint requiring it to be nonnegative. The first set of constraints we create is as follows: for each signature $\eta$, we set $$\sum_{(\alpha,\beta)\in \eta} z^\eta_{(\alpha,\beta)}=d_\eta,$$ where $d_\eta$ is the number of connected components of $G^P-X$ with signature $\eta$.

Next, for each $\{x,y\}\subseteq X$ and for $i\in [k]$, let $\Upsilon^{x,y}_i$ be the set of all configurations $(\alpha,\beta)$ such that the number of traces in $\alpha$ where $x,y$ occur consecutively is $i$.
Similarly, for $i\in [k^2]$, let $\Lambda^{x,y}_i$ be the set of all configurations $(\alpha,\beta)$ such that $\beta(\{x,y\})=i$. These sets are used to aggregate all configurations with a specific ``contribution'' $i$ to the demand and supply.
For example, all configurations whose first component is $\alpha=\{(a,b,c),(a,d,e,b),(e,b,a)\}$ would belong to $\Upsilon^{a,b}_2$, and similarly all configurations whose second component $\beta$ satisfy $\beta(\{a,b\})=3$ would belong to $\Lambda^{x,y}_3$. 
We can now add our second set of constraints: for each $\{x,y\}\subseteq X$, 
$$
\sum_{i\in[k]}\big (i\cdot \sum_{(\alpha,\beta)\in \Upsilon^{x,y}_i}z^\eta_{(\alpha,\beta)} \big )
\leq
\sum_{i\in[k^2]}\big (i\cdot \sum_{(\alpha,\beta)\in \Lambda^{x,y}_i}z^\eta_{(\alpha,\beta)}\big ).
$$

The number of variables used in the ILP instance is upper-bounded by the number of signatures times the cardinality of the largest signature. By Lemma~\ref{lem:signbounds}, the latter is upper-bounded by $k^{\bigoh(k^2)}$, and therefore the former is upper-bounded by $2^{k^{\bigoh(k^2)}}$; in total, the instance thus contains at most $2^{k^{\bigoh(k^2)}}$ variables. Since we can determine sign$(C)$ for each connected component $C$ of $G^P-X$ by Lemma~\ref{lem:signbounds} and the number of connected components is upper-bounded by $|V(G)|$, we can also construct this ILP instance in time at most $2^{k^{\bigoh(k^2)}}\cdot |V(G)|$. Moreover, Theorem~\ref{thm:pilp} allows us to solve the ILP instance in time at most $2^{2^{k^{\bigoh(k^2)}}}\cdot |V(G)|$, and in particular to output an assignment $\zeta$ from variables $z^\eta_{(\alpha,\beta)}$ to integers which satisfies the above constraints. 

To conclude the proof, we show how $\zeta$ can be used to obtain the desired configuration selector $S$. For each set of connected components with signature $\eta$, let $S$ map precisely $z^\eta_{(\alpha,\beta)}$ connected components to configuration $(\alpha,\beta)$. Observe that due to the first set of constraints, in total $S$ will be mapping precisely the correct number of components with signature $\eta$ to individual configurations $(\alpha,\beta)$. Moreover, the second set of constraints ensures that, for every $\{x,y\}\subseteq X$, the total demand does not exceed the total supply. Hence $S$ is valid and the lemma holds.
\end{proof}
}
\setcounter{THE}{3}
\begin{THE}
  \textsc{EDP} is fixed-parameter tractable parameterized by the fracture number of the augmented graph. 
\end{THE}

\begin{proof}
We begin by computing a fracture modulator of the augmented graph by Lemma~\ref{lem:computefracture}. We then use Lemma~\ref{lem:validsel} to determine whether a valid configuration selector $S$ exists, which by Lemma~\ref{lem:selcorrect} allows us to solve EDP.
\end{proof}

\lv{
\section{Hardness}\label{sec:hardness}
In the previous section we have shown that \textsc{EDP}
is fixed-parameter tractable parameterized by the fracture number.
One might be tempted to think that tractability still applies if
instead of bounding the size of each component one only bounds the
number of terminal pairs in each component. In this section we show
that this is not the case, i.e., we show that even if both the
deletion set as well as the number of terminal pairs in each component
are bounded by a
constant, \textsc{EDP} remains NP-complete. Namely, this
section is devoted to a proof of the following theorem.
}
\sv{Having established that \textsc{EDP}
is fixed-parameter tractable parameterized by the fracture number, let us briefly consider
potential extensions of the parameter.
In particular, one might be tempted to think that tractability still applies if
instead of bounding the size of each component one only bounds the
number of terminal pairs in each component. We conclude this section by showing
that this is not the case: even if both the
deletion set and the number of terminal pairs in each component
are bounded by a
constant, \textsc{EDP} remains NP-complete.
}
\setcounter{THE}{4}
\sv{\begin{THE}[$\star$]}
\lv{\begin{THE}}\label{the:EDP-NP-comp}
  \textsc{EDP} is NP-complete even if the augmented
  graph $G^P$ of the instance has a deletion set $D$ of size $6$ such that each
  component of $G^P - D$ contains at most $1$ terminal pair.
\end{THE}
\sv{The proof of Theorem~\ref{the:EDP-NP-comp} is based on a polynomial reduction from the
\textsc{Multiple Edge Disjoint Paths (MEDP)} problem, where given an
undirected graph $G$, three pairs $(s_1,t_1)$, $(s_2,t_2)$, and $(s_3,t_3)$ of terminals
and three integers $n_1$, $n_2$, and $n_3$ one asks whether there is a set
of pairwise edge disjoint paths containing $n_1$ paths between
$s_1$ and $t_1$, $n_2$ paths between $s_2$ and $t_2$, and $n_3$ paths between $s_3$ and $t_3$.
MEDP is known to be strongly NP-complete~\cite[Theorem 4]{vygen1995np}.}
\lv{We will show the theorem by a polynomial-time reduction from the
following problem. 
\pbDef{\textsc{Multiple edge disjoint Paths (MEDP)}}{An undirected
  graph $G$, three pairs $(s_1,t_1)$, $(s_2,t_2)$, and $(s_3,t_3)$ of terminals
  (vertices of $G$) and three integers $n_1$, $n_2$, and $n_3$.}{Is there a set
  of pairwise edge disjoint paths containing $n_1$ paths between
  $s_1$ and $t_1$, $n_2$ paths between $s_2$ and $t_2$, and $n_3$ paths between $s_3$ and $t_3$?}
It is shown in~\cite[Theorem 4]{vygen1995np} that MEDP is strongly NP-complete.
}

\lv{\begin{proof}[Proof of Theorem~\ref{the:EDP-NP-comp}]
  We provide a polynomial-time reduction from the \textsc{MEDP}
  problem.
  Namely, for an instance
  $\III=(G,(s_1,t_1),(s_2,t_2),(s_3,t_3),n_1,n_2,n_3)$ of \textsc{MEDP}, we construct an
  instance $\III'=(H,P)$ in polynomial time such that $\III$ has a
  solution if and only if so does $\III'$ and $H^P$ has a deletion set
  $D \subseteq V(H)$ of size $6$ such that each component of
  $H^P- D$ contains at most one terminal pair.

  The graph $H$ is obtained from $G$ after adding:
  \begin{itemize}
  \item[(C1)] $6$ vertices denoted by $s_1^D$, $t_1^D$, $s_2^D$,
    $t_2^D$, $s_3^D$, and $t_3^D$ making up the deletion set $D$,
  \item[(C2)] $n_1$ new vertices $a_1^1,\dotsc,a_1^{n_1}$ each connected via two edges to $s_1$ and
    $s_1^D$,
  \item[(C3)] $n_1$ new vertices $b_1^1,\dotsc,b_1^{n_1}$ each connected via two edges to $t_1$ and
    $t_1^D$,
  \item[(C4)] $n_2$ new vertices $a_2^1,\dotsc,a_2^{n_2}$ each connected via two edges to $s_2$ and
    $s_2^D$,
  \item[(C5)] $n_2$ new vertices $b_2^2,\dotsc,b_2^{n_2}$ each
    connected via two edges to $t_2$ and $t_2^D$,
  \item[(C6)] $n_3$ new vertices $a_3^1,\dotsc,a_3^{n_2}$ each connected via two edges to $s_3$ and
    $s_3^D$,
  \item[(C7)] $n_3$ new vertices $b_3^2,\dotsc,b_3^{n_2}$ each connected via two edges to $t_3$ and $t_3^D$,
  \item[(C8)] for every $i$ with $1 \leq i \leq n_1$ two new vertices
    $s_1^i$ and $t_1^i$, where $s_1^i$ is connected by an edge to
    $s_1^D$ and $t_1^i$ is connected by an edge to $t_1^D$,
  \item[(C9)] for every $i$ with $1 \leq i \leq n_2$ two new vertices
    $s_2^i$ and $t_2^i$, where $s_2^i$ is connected by an edge to
    $s_2^D$ and $t_2^i$ is connected by an edge to $t_2^D$.
  \item[(C10)] for every $i$ with $1 \leq i \leq n_3$ two new vertices
    $s_3^i$ and $t_3^i$, where $s_3^i$ is connected by an edge to
    $s_3^D$ and $t_3^i$ is connected by an edge to $t_3^D$.
  \end{itemize}
  This completes the description of $H$. The set $P$ is defined as
  $\SB (s_1^i,t_1^i) \SM 1 \leq i \leq n_1 \SE \cup \SB (s_2^i,t_2^i)
  \SM 1 \leq i \leq n_2 \SE\cup \SB (s_3^i,t_3^i)
  \SM 1 \leq i \leq n_3 \SE$. Observe that $H^P- D$ has one
  component consisting of the vertices $s_1^i$ and $t_1^i$ for every
  $i$ with $1 \leq i \leq n_1$, one component consisting of the
  vertices $s_2^i$ and $t_2^i$ for every $i$ with $1 \leq i \leq n_2$,
  one component consisting of the
  vertices $s_3^i$ and $t_3^i$ for every $i$ with $1 \leq i \leq n_3$,
  and one large component consisting of $G$ and the vertices
  introduced in (C2)--(C5). Hence every component of $H^P - D$
  contains at most one terminal pair.

  It remains to show that $\III$ has a solution if and only if so does
  $\III'$. For the forward direction let $S$ be a solution for $\III$,
  i.e., $S$ is a set of pairwise edge disjoint paths in $G$
  containing $n_1$ paths $P_1,\dotsc,P_{n_1}$ between $s_1$ and $t_1$,
  $n_2$ paths $Q_1,\dotsc,Q_{n_2}$ between $s_2$ and $t_2$, and
  $n_3$ paths $R_1,\dotsc,R_{n_2}$ between $s_3$ and $t_3$.
  Then we obtain a solution for $\III$, i.e., a set $S'$
  of pairwise edge disjoint paths in $H$ containing one
  path between every terminal pair in $P$ as follows. For every $i$
  with $1 \leq i \leq n_1$, we add the path between $s_1^i$ and $t_1^i$
  in $H$ defined as $(s_1^i,s_1^D,a_1^i,s_1)\concat P_i \concat
  (t_1,b_1^i,t_1^D,t_1^i)$ to $S'$, where $S_1\concat S_2$ denotes the concatenation
  of the two sequences $S_1$ and $S_2$. Similarly, for every $i$
  with $1 \leq i \leq n_2$, we add the path between $s_2^i$ and $t_2^i$
  in $H$ defined as $(s_2^i,s_2^D,a_2^i,s_2)\concat Q_i \concat
  (t_2,b_2^i,t_2^D,t_2^i)$ to $S'$.
  Finally, for every $i$
  with $1 \leq i \leq n_3$, we add the path between $s_3^i$ and $t_3^i$
  in $H$ defined as $(s_3^i,s_3^D,a_3^i,s_3)\concat R_i \concat
  (t_3,b_3^i,t_3^D,t_3^i)$ to $S'$. Because all the path added to
  $S'$ are pairwise edge disjoint and we added a path for every
  terminal pair in $P$, this shows that $S'$ is a solution for
  $\III'$.

  Towards showing the reverse direction, let $S'$ be a solution for
  $\III'$, i.e., a set of pairwise edge disjoint paths containing one
  path, denoted by $P_p$ for every terminal pair $p \in P$. Observe
  that if $p=(s_1^i,t_1^i)$ for some $i$ with $1 \leq i \leq n_1$,
  then $P_p$ contains a path between $s_1$
  and $t_1$ in $G$. The same applies if $p=(s_2^i,t_2^i)$ for some $i$
  with $1 \leq i \leq n_2$ and $p=(s_3^i,t_3^i)$ for some $i$
  with $1 \leq i \leq n_3$. Hence the set $S$
  containing the restriction of every path $P_p$ in $S'$ to $G$
  contains $n_1$ paths between $s_1$ and $t_1$, $n_2$ paths between
  $s_2$ and $t_2$ in $G$, and $n_3$ paths between $s_3$ and $t_3$,
  which are all pairwise edge disjoint. Hence
  $S$ is a solution for $\III$.
\end{proof}
}

\section{Conclusion}

Our results close a wide gap in the understanding of the complexity
landscape of EDP parameterized by structural parameterizations.
On the positive side we present three novel algorithms for the classical
EDP problem: a polynomial-time algorithm for instances with a
FVS of size one, an fpt-algorithm w.r.t. the treewidth and maximum degree of the input
graph, and an fpt-algorithm for instances that have a small
deletion set into small components in the augmented graph.
On the negative side we solve a long-standing open problem concerning
the complexity of EDP parameterized by the treewidth of the augmented
graph: unlike related multicut problems~\cite{GottlobL07}, EDP is
\W{1}-hard parameterized by the feedback vertex set number of the
augmented graph. 

\paragraph*{Acknowledgments}
Supported by the Austrian Science Fund (FWF), project P26696. Robert Ganian is also affiliated with FI MU, Brno, Czech Republic.

\bibliographystyle{plainurl}
\bibliography{bibliography}

\end{document}